\documentclass[preprint,10pt]{elsarticle}

\biboptions{sort&compress}
\usepackage{epsfig}
\usepackage{graphicx}
  \usepackage{epsfig,epsf,latexsym,subfigure,epstopdf}
\makeatletter
\newcommand*\bigcdot{\mathpalette\bigcdot@{.5}}
\newcommand*\bigcdot@[2]{\mathbin{\vcenter{\hbox{\scalebox{#2}{$\m@th#1\bullet$}}}}}
\makeatother
\usepackage{amssymb}
\usepackage{amsthm}
\usepackage{graphicx}

\newcommand{\bs}{\boldsymbol}

	\usepackage{amsmath}
    \usepackage{color}
	\usepackage{amsthm}
	\usepackage{amsfonts}
\usepackage{epsfig,epsf,latexsym,subfigure,epstopdf}
    \usepackage{float}
	\usepackage{nicefrac}
	\usepackage{url}
	\usepackage[toc,page]{appendix}

	\newtheorem{thm}{Theorem}[section]

	\newdefinition{rmk}[thm]{Remark}

	\newproof{pf}{Proof}
	\allowdisplaybreaks
\journal{Applied Numerical Mathematics}

\begin{document}
\begin{frontmatter}

\title{A Thermodynamically Consistent Phase-Field Model and an Entropy Stable Numerical Method for Simulating Two-Phase Flows with Thermocapillary Effects\tnoteref{mytitlenote}}

%\author{Yanxiao Sun \fnref{myfootnote}}
%\address{Mechanics Division, Beijing Computational Science Research Center, Building 9, East Zone, ZPark II, No. 10 East Xibeiwang Road, Haidian District, Beijing 100193, People鈥檚 Republic of China}

\author[mymainaddress]{Yanxiao Sun}
\author[mysecondaryaddress]{Jiang Wu}
\author[mythirdaddress]{Maosheng Jiang}
\author[myfourthaddress]{Steven M. Wise}
\author[mymainaddress]{Zhenlin Guo \corref{mycorrespondingauthor}}
\cortext[mycorrespondingauthor]{Corresponding author}
\ead{zguo@csrc.ac.cn}

\address[mymainaddress]{Mechanics Division, Beijing Computational Science Research Center, Building 9, East Zone, ZPark II, No. 10 East Xibeiwang Road, Haidian District, Beijing 100193, People Republic of China}

\address[mysecondaryaddress]{School of Mathematics and Physics, University of Science and Technology Beijing, No.30, Xueyuan Road, Haidian District, Beijing 100083, People Republic of China}

\address[mythirdaddress]{School of Mathematics and Statistics, Qingdao University, Qingdao 266071, People Republic of China}

\address[myfourthaddress]{Mathematics Department, University of Tennessee, Knoxville TN 37996-1320, USA}

    \begin{abstract}
In this study, we have derived a thermodynamically consistent phase-field model for two-phase flows with thermocapillary effects. This model accommodates variations in physical properties such as density, viscosity, heat capacity, and thermal conductivity between the two components. The model equations encompass a Cahn-Hilliard equation with the volume fraction as the phase variable, a Navier-Stokes equation, and a heat equation, and meanwhile maintains mass conservation, energy conservation, and entropy increase simultaneously. Given the highly coupled and nonlinear nature of the model equations, we developed a semi-decoupled, mass-preserving, and entropy-stable time-discrete numerical method. We conducted several numerical tests to validate both our model and numerical method. Additionally, we have investigated the merging process of two bubbles under non-isothermal conditions and compared the results with those under isothermal conditions. Our findings reveal that temperature gradients influence bubble morphology and lead to earlier merging. Moreover, we have observed that the merging of bubbles slows down with increasing heat Peclect number $\mathrm{Pe}_{T}$ when the initial temperature field increases linearly along the channel, while bubbles merge faster with heat Peclect number $\mathrm{Pe}_{T}$ when the initial temperature field decreases linearly along the channel.
    \end{abstract}


\end{frontmatter}

\section{Introduction}

 The variations of surface tension caused by temperature gradients at a fluid-fluid interface usually lead to an interfacial shear force along the interface, and thus induce the movement of fluids in the direction of the temperature gradient. This effect is known as the thermocapillary effect, which plays an important role in various industrial applications involving microgravity or microdevices~\cite{guo2015thermodynamically}. Several phase-field models have been developed for simulating the thermocapillary effects for two-phase flows \cite{guo2015thermodynamically, liu2012modeling, liu2013phase, mitchell2021computational, hu2019diffuse, xiao2022spectral, yang2022phase,liu2014lattice, yue2022improved, sun2020structure}. The essential idea for the phase-field model is to introduce an order parameter to characterize the different phases, which varies continuously over a thin interfacial layer and is essentially uniform in the bulk phases. However, most of the existing models are not thermodynamically consistent, namely the fluid flow equation (Navier-Stokes), the phase-field equation (Cahn-Hilliard), and the heat equation are simply coupled and thus do not satisfy energy conservation and entropy production laws.

Recently, a thermodynamically consistent phase-field model \cite{guo2015thermodynamically} has been developed for simulating two-phase flows with thermocapillary effects, where the mass concentration is employed as the phase variable. The model equations for the whole computational domain can be derived variationally from energy and entropy functional, which allows the two fluids to have different physical properties (including density, viscosity and thermal conductivity) and meanwhile maintains mass conservation, internal energy conservation, and entropy increase. The model equations are highly nonlinear and coupled, which leave a challenge to the numerical simulations. In \cite{sun2020structure}, another thermodynamically consistent phase-field model was developed for simulating thermocapillary effects, where the volume fraction is employed as the phase variable instead of the mass concentration, and the two components are assumed to be of equal density. Moreover, several stable and efficient numerical methods, including IEQ \cite{yang2016linear, yang2017efficient, yang2022efficient}, SAV \cite{shen2018scalar, shen2019new, zhu2019efficient, wang2021second}, extended SAV \cite{HouXu1, HouXu2,Jiang1, hou2019variant} and SVM \cite{GongHongWang,HongWangGong,sun2020structure}, have been developed for solving the phase-field models. In some of these methods, extra, auxiliary variables have been introduced to ensure the discrete energy conservation and entropy increase.

In this paper, we propose a thermodynamically consistent phase-field model for two-phase flows with thermocapillary effects, which allows the two components to have different physical properties, including density, viscosity, heat capacity, and thermal conductivity. The model equations consist of a Cahn-Hilliard equation with the volume fraction of one component as the phase variable, a Navier-Stokes equation, and a heat equation. These equations are highly coupled and nonlinear, and meanwhile satisfy mass conservation, total energy conservation, and entropy increase. To carry out the numerical simulations, we develop a first order, semi-decoupled, mass conservative and entropy stable numerical method for solving the model equations, where there is no need to introduce extra auxiliary variables to our method.

The paper is organized as follows. We present the derivations of the phase-field model for two-phase flows with thermocapillary effects in \S\ref{section-2}, and provide the model validation in \S\ref{section-3}.  The numerical method and the corresponding discrete mass conservation and entropy increase are shown in \S\ref{section-discrete-numerical-method}. The numerical results are presented in \S\ref{numerical result}, and finally the conclusion is given in \S\ref{conclusion}.

\section{Model equations} \label{section-2}
    \subsection{Phase-field variable and variable physical properties}
We use the phase-field model to represent a two-phase incompressible fluid flow with variable physical properties and thermocapillary effects along the fluid/fluid interface.  In particular, we use the following formulation as the variable density for the two-phase fluid:
    \begin{align}
\rho\left(\Phi_{1},\Phi_{2}\right)=\rho_{1}\Phi_{1} + \rho_{2}\Phi_{2},
    \end{align}
where $\rho_{i}>0$ is the constant density and $\Phi_{i}$ is the volume fraction of the $i^\mathrm{th}$ fluid, respectively, for $i = 1,2$. Here, we use the no-voids assumption,
    \begin{align}
\Phi_{1}+\Phi_{2}=1.
    \end{align}
We next treat the two-phase fluid as one mixture, and define the phase variable as
    \[
\psi=\Phi_{1},
    \]
such that the variable density $\rho$ for the mixture can be rewritten as
    \begin{align}
\rho=\rho_{1}\psi + \rho_{2}(1-\psi).  \label{defininition_density}
    \end{align}
Similarly, we define the other variable properties for the mixture:
    \begin{align}
{\rm variable~viscosity:}&~~\mu(\psi)=\mu_{1} \psi +\mu_{2}(1-\psi),
    \label{def_variable-vis}
    \\
{\rm variable~thermal~conductivity:}&~~k(\psi)=k_{1} \psi +k_{2}(1-\psi),
    \label{def_variable-k}
    \\
{\rm variable~heat~capacity:}&~~C_h(\psi)=C_{h,1} \psi +C_{h,2}(1-\psi).
    \label{def_variable-chc}
    \end{align}

\subsection{Internal energy, free energy, and entropy}

To investigate the thermocapillary effects, we expect the surface free energy of the two-phase fluid to be temperature-dependent. To address this, we propose the internal energy density, $\hat{u}$, free energy density, $\hat{f}$, and entropy density, $\hat{s}$, for the mixture as follows:
    \begin{align}
\hat{u}(T, \psi, \boldsymbol{\nabla} \psi)&=u(T, \psi)+\lambda_{u} \delta(\psi, \boldsymbol{\nabla} \psi),
    \label{definition_hat_u}
    \\
\hat{f}(T, \psi, \boldsymbol{\nabla} \psi)&=f(T, \psi)+\lambda_{f} \delta(\psi, \boldsymbol{\nabla} \psi),
    \\
    \label{definition_hat_s}	
\hat{s}(T, \psi, \boldsymbol{\nabla} \psi)&=s(T, \psi)+\lambda_{s} \delta(\psi, \boldsymbol{\nabla} \psi).
    \end{align}
Here $u$, $f$ and $s$ represent the classical components that can be defined as~\cite{guo2015thermodynamically}
    \begin{align}
u &= \rho C_h T,
    \\
f &= \rho C_hT-\rho C_hT \ln \left(\frac{T}{T_0}\right),
    \\
s &= \rho C_h \ln \left(\frac{T}{T_0}\right),
    \end{align}
where $T$ denotes the absolute temperature, and $T_{0}$ serves as the reference temperature. Furthermore, we note the following thermodynamic relation:
    \begin{equation}
     f = u - Ts. \label{the-relationship-of-u-and-f}
    \end{equation}
Additionally, $\delta$ represents the interface free energy, defined as:
    \begin{equation}
    \label{definition_of_delta}
\delta(\psi, \boldsymbol{\nabla} \psi) =\frac{1}{\epsilon} W(\psi)+\epsilon \frac{|\boldsymbol{\nabla} \psi|^2}{2},
    \end{equation}
where $W(\psi)=\psi^2(1-\psi)^2/4$ represents a double-well free energy, and $\epsilon$ serves as a small parameter denoting the thickness of the diffuse interface. 

In the framework of the sharp-interface model, it's common to assume that the surface tension decreases linearly with temperature, given by:
    \begin{equation}
\sigma(T) =\sigma_{0}-\sigma_{T}(T-T_{0}),
    \label{surface_tension_T}
    \end{equation}
where $\sigma_{0}$ represents the surface tension at the reference temperature $T_0$, and $\sigma_{T}$ denotes the rate of change of surface tension with temperature. Consequently, for our phase-field model, we consider $\lambda_u$, $\lambda_f$, and $\lambda_s$ as:
    \begin{equation}
\lambda_{u}= \eta (\sigma_{0}+\sigma_{T}T_{0}),\quad \lambda_{f}= \eta  \sigma(T), \quad \lambda_{s}= \eta \sigma_T,
    \label{surface-tension}
    \end{equation}
where $\eta > 0$ is a positive constant that establishes the relationship between the surface tensions of the diffuse-interface model and the sharp-interface model, as discussed in \cite{guo2015thermodynamically}. Moreover, we note the following thermodynamic relations 
    \begin{align}
\lambda_f &= \lambda_u - T\lambda_s, \label{relationship-of-lambdau}\\
\hat{f} &= \hat{u}-T \hat{s}.\label{the-relationship-of-hatu-and-hatf}
    \end{align}

\subsection{Conservation laws}
Next, we establish the conservation laws for the flow of the two-phase fluid. Assuming the two-phase fluid occupies a domain $\Omega$, we consider an arbitrary material volume $V(t)\in \Omega$ moving with the mixture. Here, we define the following quantities: the total mass $M$, volume of a single fluid phase $\Psi$, momentum $\boldsymbol{P}$, and internal energy $U$:
\begin{align}
M&=\int_{V(t)}\rho~dV,\label{def-mass}\\
\Psi&=\int_{V(t)}\psi~dV,\label{def-phase}\\
{\boldsymbol{P}}&=\int_{V(t)}\rho\boldsymbol{v}~dV,\label{def-mom}\\
U&=\int_{V(t)} \hat{u} ~dV,\label{def-energy}
\end{align}
where $\rho$ is defined in Eq.~(\ref{defininition_density}), $\boldsymbol{v}$ represents the mass-averaged velocity of the mixture as discussed in \cite{lowengrub1998quasi}, and $\hat{u}$ denotes the internal energy density as defined in Eq.~(\ref{definition_hat_u}). Given these considerations, the associated conservation laws can be expressed as follows:
\begin{align}
\frac{dM}{d t} &= 0, \label{def-mass-laws} \\
\frac{d\Psi}{dt} &=  \int_{\partial V(t)} -\frac{\boldsymbol{J}}{\rho_{1}} \cdot \boldsymbol{n}~dA,\label{def-phase-laws}\\
\frac{d\boldsymbol{P}}{ dt} &=\int_{\partial V(t)}  \mathbf{T}\cdot \boldsymbol{n}~dA-\int_{V(t)}\rho g \hat{\boldsymbol{z}}~dV,\label{def-mom-laws}\\
\frac{dU}{dt}&=\int_{\partial V(t)} -\boldsymbol{q}_{E}\cdot \boldsymbol{n} ~dA+\int_{V(t)}\boldsymbol{\nabla} \boldsymbol{v}: \mathbf{T} ~ d {V}.\label{def-energy-laws}
\end{align}
Here, $\boldsymbol{n}$ denotes the unit outward normal vector of the boundary $\partial V(t)$, while $\hat{\boldsymbol{z}}$ represents the unit vector in the vertical direction. Eq.~\eqref{def-mass-laws} represents the mass conservation of the mixture within the material volume, and Eq.~\eqref{def-phase-laws} stands for the volume conservation for a single fluid (here for the fluid 1), where $\boldsymbol{J}$ is the volume flux of fluid 1 through the boundary of the material volume. The momentum conservation Eq.~\eqref{def-mom-laws} states that the rate of change in total momentum equals the force (surface forces $\mathbf{T}$) acting on the volume boundary. Here we assume that
\begin{align} \label{defintion_T}
	\mathbf{T}= \tau-p \mathbf{I}+ \mathbf{q},
\end{align}
where $\mathbf{\tau}=\mathcal{\mu}\left(\boldsymbol{\nabla} \boldsymbol{v}+\boldsymbol{\nabla} \boldsymbol{v}^{\top}\right)-\frac{2}{3} \mathcal{\mu}(\boldsymbol{\nabla} \cdot \boldsymbol{v}) \mathbf{I}$ represents the deviatoric stress tensor, $p$ denotes the pressure, $\mathbf{I}$ stands for the unit tensor, and $\mathbf{q}$ represents the unknown term that needs to be determined to ensure entropy non-decrease. Additionally, $g$ in Eq.~\eqref{def-mom-laws} corresponds the gravitational constant. The internal energy conservation Eq.~\eqref{def-energy-laws} states that the change in internal energy equals the rate of work  done by the forces ($\mathbf{T}$) on the boundary plus the energy flux ($\boldsymbol{q}_{E}$) through the volume boundary.
%The derivation of $\mathbf{q}$ is shown in \ref{appendix_Deriv_entropy generation}

\subsection{Model derivation}
By substituting Eqs.~(\ref{def-mass}) -- (\ref{def-energy}) into the conservation laws (\ref{def-mass-laws}) -- (\ref{def-energy-laws}), we derive the following equations:
    \begin{align}
\rho \frac{D \boldsymbol{v}}{D t}&=\boldsymbol{\nabla} \cdot \mathbf{T}-\rho g \hat{\boldsymbol{z}},
    \label{momentum_cons_unknown}
    \\
\boldsymbol{\nabla} \cdot \boldsymbol{v}&=-\alpha \frac{\boldsymbol{\nabla} \cdot \boldsymbol{J}}{\rho_{1}},
    \label{mass_cons_unknown}
    \\
\frac{\partial \psi}{\partial t}+\boldsymbol{\nabla} \cdot(\psi \boldsymbol{v})&=-\frac{\boldsymbol{\nabla} \cdot \boldsymbol{J}}{\rho_{1}} \quad \mbox{or}\quad \frac{D \psi}{D t}+\psi(\boldsymbol{\nabla} \cdot \boldsymbol{v})=\frac{-\boldsymbol{\nabla} \cdot \boldsymbol{J}}{\rho_{1}},
    \label{volume_cons_unknown}
    \\
\frac{\partial u}{\partial T} \frac{D T}{D t}&=-\boldsymbol{\nabla} \cdot \boldsymbol{q}_{E}+\boldsymbol{\nabla} \boldsymbol{v}: \mathbf{T}-src,
    \label{energy_cons_unknown}
    \end{align}
where
    \begin{align}
src&=\left(\frac{\partial u}{\partial \psi}+\lambda_{u} w\right) \frac{D \psi}{D t}+\lambda_{u} \epsilon \boldsymbol{\nabla} \cdot (\boldsymbol{\nabla} \psi \frac{\partial \psi}{\partial t})+\hat{u}(\boldsymbol{\nabla} \cdot \boldsymbol{v}) \nonumber
    \\
& \quad +\lambda_{u} \epsilon \boldsymbol{\nabla}\cdot (\boldsymbol{\nabla} \psi\otimes \boldsymbol{\nabla} \psi)\cdot \boldsymbol{v},
    \\
w&=\frac{W^{\prime}(\psi)}{\epsilon} -\epsilon \Delta \psi,\\
\alpha&={(\rho_{2}-\rho_{1})}/{ \rho_{2}}, \label{definition_alpha-nodim}
    \end{align}
and $D/Dt=\partial /\partial t+\boldsymbol{v} \cdot \boldsymbol{\nabla}$ denotes the material derivative. Note that the comprehensive derivation is provided in \ref{derivation_of_coservation_equa}.

We now define the entropy as
\begin{align}
S&=\int_{V(t)} \hat{s}  ~dV,
\end{align}
and calculate the time derivative of the entropy. Here $\hat{s}$ represents the entropy density as defined in Eq.~(\ref{definition_hat_s}), and the following relation is used
\begin{align} \label{partial_u_T}
\frac{\partial s}{\partial T}\left(\frac{\partial u}{\partial T}\right)^{-1}=\frac{1}{T}.
\end{align}
Similar to the derivation of Eq.~(\ref{energy_conser_equa_2}), we obtain
\begin{align}
\label{entropy_derivative_0}	\frac{\partial \hat{s}}{\partial t} &=\frac{\partial s}{\partial T} \frac{\partial T}{\partial t}+\left(\frac{\partial s}{\partial \psi}+\lambda_{s} w\right) \frac{\partial \psi}{\partial t}+\lambda_{s} \epsilon \boldsymbol{\nabla} \cdot (\boldsymbol{\nabla} \psi \frac{\partial \psi}{\partial t}),
\end{align}
such that
\begin{align}
\label{entropy_derivative}	\frac{d S}{d t} =&\int_{V(t)}\left\{\frac{\partial s}{\partial T} \frac{\partial T}{\partial t}+\left(\frac{\partial s}{\partial \psi}+\lambda_{s} w\right) \frac{\partial \psi}{\partial t}+\lambda_{s} \epsilon \boldsymbol{\nabla} \cdot (\boldsymbol{\nabla} \psi \frac{\partial \psi}{\partial t})+\boldsymbol{\nabla} \cdot({\hat{s}\boldsymbol{v}})\right\} d {V} \notag\\
	=&\int_{V(t)}\bigg\{\frac{1}{T} \frac{\partial u}{\partial T} \frac{\partial T}{\partial t}+\left(\frac{\partial s}{\partial \psi}+\lambda_{s} w\right) \frac{\partial \psi}{\partial t}+\lambda_{s} \epsilon \boldsymbol{\nabla} \cdot (\boldsymbol{\nabla} \psi \frac{\partial \psi}{\partial t})+\boldsymbol{\nabla} \cdot({\hat{s}\boldsymbol{v}})\bigg\} d{V}.
\end{align}
Substituting (\ref{energy_cons_unknown}) into (\ref{entropy_derivative}), we obtain
\begin{align} \label{entropy_derivative_1}
	\frac{d S}{d t}=&\int_{V(t)}\left\{-  \frac{\partial s}{\partial T} \boldsymbol{v} \cdot \boldsymbol{\nabla} T-\frac{1}{T}\left(\frac{\partial u}{\partial \psi}+\lambda_{u} w\right) \frac{D \psi}{D t}-\frac{1}{T}\lambda_{u} \epsilon \boldsymbol{\nabla} \cdot (\boldsymbol{\nabla} \psi \frac{\partial \psi}{\partial t})
	\right. \notag\\
&-\frac{1}{T}\hat{u}(\boldsymbol{\nabla} \cdot \boldsymbol{v})-\frac{1}{T}\boldsymbol{\nabla} \cdot \boldsymbol{q}_{E}+\frac{1}{T}\boldsymbol{\nabla} \boldsymbol{v}: \mathbf{T}-\frac{1}{T}\lambda_{u} \epsilon \boldsymbol{\nabla}\cdot (\boldsymbol{\nabla} \psi\otimes \boldsymbol{\nabla} \psi)\cdot \boldsymbol{v} \notag\\
&\left.+\left(\frac{\partial s}{\partial \psi}+\lambda_{s} w\right)\left(\frac{D \psi}{D t}-\boldsymbol{v}\cdot\boldsymbol{\nabla} \psi\right)+\lambda_{s} \epsilon \boldsymbol{\nabla} \cdot (\boldsymbol{\nabla} \psi \frac{\partial \psi}{\partial t})+\boldsymbol{\nabla} \cdot({\hat{s}\boldsymbol{v}})\right\}d {V}.
 \end{align}
With the help of Eq.~(\ref{relationship-of-lambdau}) and the following identities
 \begin{align}
 \frac{1}{T}\frac{\partial f}{\partial \psi}&=\frac{1}{T}\frac{\partial u}{\partial \psi}-\frac{\partial s}{\partial \psi},\label{partial_u_phi}\\
\boldsymbol{v}\cdot \boldsymbol{\nabla} \hat{s}&=\boldsymbol{v}\cdot\left(\frac{\partial s}{\partial T} \boldsymbol{\nabla} T+\left(\frac{\partial s}{\partial \psi}+\lambda_{s} w \right) \boldsymbol{\nabla} \psi +\lambda_{s} \epsilon \boldsymbol{\nabla}\cdot (\boldsymbol{\nabla} \psi\otimes \boldsymbol{\nabla} \psi)\right),  \label{nabla_s}
 \end{align}
Eq.~(\ref{entropy_derivative_1}) can be rewritten as
   \begin{align} \label{entropy_derivative_02}
   \frac{d S}{d t} =&\int_{V(t)}\left\{-\frac{1}{T} \boldsymbol{\nabla} \cdot \boldsymbol{q}_{E}+\frac{1}{T} \boldsymbol{\nabla} \boldsymbol{v}:\mathbf{T} -\frac{1}{T} \lambda_{u} \epsilon \boldsymbol{\nabla}\cdot (\boldsymbol{\nabla} \psi\otimes \boldsymbol{\nabla} \psi)\cdot \boldsymbol{v}  \right. \notag\\
  &-\frac{1}{T}\lambda_{u} \epsilon \boldsymbol{\nabla} \cdot (\boldsymbol{\nabla} \psi \frac{\partial \psi}{\partial t})+\hat{s}(\boldsymbol{\nabla}\cdot\boldsymbol{v})   + \lambda_{s} \epsilon \boldsymbol{\nabla}\cdot (\boldsymbol{\nabla} \psi\otimes \boldsymbol{\nabla} \psi)\cdot \boldsymbol{v} \notag\\
&\left. -\frac{1}{T} \hat{u}(\boldsymbol{\nabla} \cdot \boldsymbol{v})-\frac{1}{T}\left(\frac{\partial f}{\partial \psi}+\lambda_{f} w\right) \frac{D \psi}{D t}+\lambda_{s} \epsilon \boldsymbol{\nabla} \cdot (\boldsymbol{\nabla} \psi \frac{\partial \psi}{\partial t})     \right\} d {V}.
\end{align}
In addition, with the help of Eq.~(\ref{the-relationship-of-hatu-and-hatf}) and the following identity
\begin{align}
\boldsymbol{\nabla}\cdot (\boldsymbol{\nabla} \psi\otimes \boldsymbol{\nabla} \psi)\cdot \boldsymbol{v}&=\boldsymbol{\nabla}\cdot \left((\boldsymbol{\nabla} \psi\otimes \boldsymbol{\nabla} \psi)\cdot \boldsymbol{v}\right)-\boldsymbol{\nabla} \boldsymbol{v}:(\boldsymbol{\nabla} \psi\otimes \boldsymbol{\nabla} \psi),
\end{align}
we rewrite Eq.~(\ref{entropy_derivative_02}) as
 \begin{align}
 \frac{d S}{d t}  =&\int_{V(t)}\bigg\{-\frac{1}{T} \boldsymbol{\nabla} \cdot \left(\boldsymbol{q}_{E}+\lambda_{u} \epsilon (\boldsymbol{\nabla} \psi\otimes \boldsymbol{\nabla} \psi)\cdot \boldsymbol{v}+\lambda_{u} \epsilon \boldsymbol{\nabla} \psi \frac{\partial \psi}{\partial t}\right) \notag\\
 &+\frac{1}{T} \boldsymbol{\nabla} \boldsymbol{v}:\left(\mathbf{T}+\lambda_{f} \epsilon (\boldsymbol{\nabla} \psi\otimes \boldsymbol{\nabla} \psi)\right)-\frac{1}{T} \hat{f}(\boldsymbol{\nabla} \cdot \boldsymbol{v})  \notag\\
 &+ \lambda_{s} \epsilon \boldsymbol{\nabla}\cdot \left((\boldsymbol{\nabla} \psi\otimes \boldsymbol{\nabla} \psi)\cdot \boldsymbol{v}\right)-\frac{1}{T}\left(\frac{\partial f}{\partial \psi}+\lambda_{f} w\right) \frac{D \psi}{D t} \notag\\
 &+\lambda_{s} \epsilon \boldsymbol{\nabla} \cdot (\boldsymbol{\nabla} \psi \frac{\partial \psi}{\partial t}) \bigg\} d {V}.  \label{entropy_derivative_04}
   \end{align}
Next, by substituting Eq.~(\ref{defintion_T}) (the definition of $\mathbf{T}$), mass conservation Eq.~(\ref{mass_cons_unknown}) and volume conservation Eq.~(\ref{volume_cons_unknown}) into Eq.~(\ref{entropy_derivative_04}), we obtain
 \begin{align} \label{entropy_derivative_05}
	\frac{d S}{d t}=&\int_{V(t)}\bigg\{-\frac{1}{T} \boldsymbol{\nabla} \cdot \left(\boldsymbol{q}_{E}+\lambda_{u} \epsilon (\boldsymbol{\nabla} \psi\otimes \boldsymbol{\nabla} \psi)\cdot \boldsymbol{v}+\lambda_{u} \epsilon \boldsymbol{\nabla} \psi \frac{\partial \psi}{\partial t}\right) \notag\\
&+\frac{1}{T} \boldsymbol{\nabla} \boldsymbol{v}:\left(\mathbf{\tau}+\mathbf{q}+\lambda_{f} \epsilon (\boldsymbol{\nabla} \psi\otimes \boldsymbol{\nabla} \psi)+\mu_{0}\psi \mathbf{I}\right) \notag\\
 &-\frac{1}{T} \hat{f}(\boldsymbol{\nabla} \cdot \boldsymbol{v})+\frac{1}{T}\mu_{c} \boldsymbol{\nabla} \cdot \frac{\boldsymbol{J}}{\rho_{1}}+ \lambda_{s} \epsilon \boldsymbol{\nabla}\cdot \left((\boldsymbol{\nabla} \psi\otimes \boldsymbol{\nabla} \psi)\cdot \boldsymbol{v}\right) \notag\\
 &+\lambda_{s} \epsilon \boldsymbol{\nabla} \cdot (\boldsymbol{\nabla} \psi \frac{\partial \psi}{\partial t})\bigg\} d {V},
\end{align}
where
\begin{align}
\mu_{0}&=\frac{\partial f}{\partial \psi}+\lambda_{f} w,\\
\mu_{c}&=\mu_{0}+\alpha p.
\end{align}
With the help of the following identity
\begin{align}
\frac{1}{T} \mu_{c} \boldsymbol{\nabla} \cdot  \frac{\boldsymbol{J}}{\rho_{1}}=\boldsymbol{\nabla} \cdot \left(\frac{1}{T}\mu_{c} \frac{\boldsymbol{J}}{\rho_{1}}\right)-\mu_{c}\frac{\boldsymbol{J}}{\rho_{1}}\cdot \boldsymbol{\nabla}\left(\frac{1}{T}\right)-\frac{1}{T}\frac{\boldsymbol{J}}{\rho_{1}}\cdot \boldsymbol{\nabla} \mu_{c},
\end{align}
Eq.~(\ref{entropy_derivative_05}) can be rewritten as
\begin{align}
	\frac{d S}{d t}
 =&\int_{V(t)}\left\{  - \boldsymbol{\nabla} \cdot  \left(\frac{1}{T}\boldsymbol{q}_{E}+\frac{1}{T}\lambda_{f} \epsilon (\boldsymbol{\nabla} \psi\otimes \boldsymbol{\nabla} \psi)\cdot \boldsymbol{v}+\frac{1}{T}\lambda_{f} \epsilon \boldsymbol{\nabla} \psi \frac{\partial \psi}{\partial t}-\frac{1}{T}\mu_{c} \frac{\boldsymbol{J}}{\rho_{1}} \right) \right. \notag\\
 &+ \boldsymbol{\nabla} \left(\frac{1}{T}\right) \cdot \left(\boldsymbol{q}_{E}+\lambda_{u} \epsilon (\boldsymbol{\nabla} \psi\otimes \boldsymbol{\nabla} \psi)\cdot \boldsymbol{v}+\lambda_{u} \epsilon \boldsymbol{\nabla} \psi \frac{\partial \psi}{\partial t}-\mu_{c}\frac{\boldsymbol{J}}{\rho_{1}}\right) \notag\\
&+\frac{1}{T} \boldsymbol{\nabla} \boldsymbol{v}:\left(\mathbf{\tau}+\mathbf{q}+\lambda_{f} \epsilon (\boldsymbol{\nabla} \psi\otimes \boldsymbol{\nabla} \psi)+\mu_{0}\psi \mathbf{I}-\hat{f}\mathbf{I}\right) \notag\\
	 &\left.-\frac{1}{T}\frac{\boldsymbol{J}}{\rho_{1}}\cdot \boldsymbol{\nabla} \mu_{c}\right\} d {V}.
\end{align}
For the above equation, we denote the first part as the entropy flux of the material volume, and the rest parts as the local entropy increase, $S_{inc}$, which needs to be non-negative according the second law of thermodynamics, such that
\begin{align}
S_{inc}&=\int_{V(t)}\bigg\{\boldsymbol{\nabla} \left(\frac{1}{T}\right) \cdot \left(\boldsymbol{q}_{E}+\lambda_{u} \epsilon (\boldsymbol{\nabla} \psi\otimes \boldsymbol{\nabla} \psi)\cdot \boldsymbol{v}+\lambda_{u} \epsilon \boldsymbol{\nabla} \psi \frac{\partial \psi}{\partial t}-\mu_{c}\frac{\boldsymbol{J}}{\rho_{1}}\right) \notag\\
&+\frac{1}{T} \boldsymbol{\nabla} \boldsymbol{v}:\left(\mathbf{\tau}+\mathbf{q}+\lambda_{f} \epsilon (\boldsymbol{\nabla} \psi\otimes \boldsymbol{\nabla} \psi)+\mu_{0}\psi \mathbf{I}-\hat{f}\mathbf{I}\right)-\frac{1}{T}\frac{\boldsymbol{J}}{\rho_{1}}\cdot \boldsymbol{\nabla} \mu_{c} \bigg\} \ge 0. \label{Sgeneration}
\end{align}
To comply with the second law of thermodynamics, which states that local entropy generation must be non-negative for an irreversible process, we specify the unknown terms as:
\begin{align}
\mathbf{q}=&\hat{f}\mathbf{I}-\lambda_{f} \epsilon (\boldsymbol{\nabla} \psi\otimes \boldsymbol{\nabla} \psi)-\mu_{0}\psi \mathbf{I},\label{stress_tensor}\\
%\mathbf{\tau}=&\mathcal{\mu}\left(\boldsymbol{\nabla} \boldsymbol{v}+\boldsymbol{\nabla} \boldsymbol{v}^{\top}\right)-\frac{2}{3} \mathcal{\mu}(\boldsymbol{\nabla} \cdot \boldsymbol{v}) \mathbf{I},\\
\boldsymbol{J}=&-\rho_{1}m_{\psi}\boldsymbol{\nabla} \mu_{c},\\ %~~m_{\psi}:~m^{3}s/kg
\boldsymbol{q}_{E}=&-k\boldsymbol{\nabla} T-\lambda_{u} \epsilon (\boldsymbol{\nabla} \psi\otimes \boldsymbol{\nabla} \psi)\cdot \boldsymbol{v}-\lambda_{u} \epsilon \boldsymbol{\nabla} \psi \frac{\partial \psi}{\partial t}+\mu_{c}\frac{\boldsymbol{J}}{\rho_{1}}, \label{energy_flux}
\end{align}
where $m_{\psi}=m\sqrt{\psi^2(1-\psi)^2}$ is the degenerate mobility for the diffuse interface, and $m$ is a constant.
%$\mu$ and $k$ are the variable viscosity and thermal conductivity  for the two-phase flows, respectively,
Substituting Eqs.~(\ref{stress_tensor}) -- (\ref{energy_flux}) into Eq.~(\ref{Sgeneration}), we obtain the entropy generation for the two-phase fluid system
\begin{equation}
	\begin{aligned}
	S_{inc}=\int_{V(t)}\bigg\{  k\frac{|\boldsymbol{\nabla} T|^2}{T^2}+\frac{1}{T} \boldsymbol{\nabla} \boldsymbol{v}: \mathbf{\tau}+\frac{1}{T}m_{\psi}|\boldsymbol{\nabla} \mu_{c}|^2\bigg\} \geq0.
	\end{aligned}
\end{equation}

Substituting Eqs.~(\ref{stress_tensor}) -- (\ref{energy_flux}) into Eqs.~(\ref{momentum_cons_unknown}) -- (\ref{energy_cons_unknown}), respectively, we finally obtain the model equations for two-phase flow with variable properties and thermocapillary effects:
\begin{align}
\rho \frac{\partial \boldsymbol{v}}{\partial t}+\rho\boldsymbol{v}\cdot\boldsymbol{\nabla} \boldsymbol{v}=&\boldsymbol{\nabla} \cdot \mathbf{T}-\rho g \hat{\boldsymbol{z}},\label{model_momentum-conservation}\\
\boldsymbol{\nabla} \cdot \boldsymbol{v}=&\alpha \boldsymbol{\nabla} \cdot(m_{\psi} \boldsymbol{\nabla}\mu_{c}),\label{model_mass-conservation}\\
\frac{\partial \psi}{\partial t}+\boldsymbol{\nabla} \cdot(\psi \boldsymbol{v}) =&\boldsymbol{\nabla} \cdot(m_{\psi} \boldsymbol{\nabla}\mu_{c}), \label{model_volume_conservation}\\
\frac{\partial \left(\rho C_h T\right)}{\partial t}+\boldsymbol{\nabla}\cdot (\rho C_h T \boldsymbol{v}) =& -\boldsymbol{\nabla} \cdot \boldsymbol{q}_{E}+\boldsymbol{\nabla}\boldsymbol{v}:\mathbf{T}
    \nonumber
    \\
& - \lambda_{u}\left(\frac{\partial \delta}{\partial t}+\boldsymbol{\nabla}\cdot (\delta \boldsymbol{v})\right), \label{model_energy_conservation}
\end{align}
where
    \begin{align}
\Lambda=&\Lambda_{1} \psi+\Lambda_{2}(1-\psi)~~{\rm for} ~~\Lambda=\rho, \mu, C_h,k, \label{definition_rho}\\
\mathbf{T}=&-p \mathbf{I}+ \hat{f} \mathbf{I}+\mathbf{\tau} -\lambda_{f}(T) \epsilon (\boldsymbol{\nabla} \psi\otimes \boldsymbol{\nabla} \psi)-\mu_{0}\psi \mathbf{I},\\
\hat{f}=&f+\lambda_{f}(T) \delta,~~~
f = \rho C_hT-\rho C_hT \ln \left(\frac{T}{T_0}\right),\\
\delta=&\frac{W(\psi)}{\epsilon} +\epsilon \frac{|\boldsymbol{\nabla} \psi|^2}{2}=\frac{\psi^2(1-\psi)^2}{4\epsilon}+\epsilon \frac{|\boldsymbol{\nabla} \psi|^2}{2},
    \label{definition_delta_n+1}
    \\
\mathbf{\tau}=&\mathcal{\mu}\left(\boldsymbol{\nabla} \boldsymbol{v}+\boldsymbol{\nabla} \boldsymbol{v}^{\top}\right)-\frac{2}{3} \mathcal{\mu}(\boldsymbol{\nabla} \cdot \boldsymbol{v}) \mathbf{I},\\
\mu_{0}=&\frac{\partial f}{\partial \psi}+\lambda_{f} w, \label{continue_mu_0}\\
\frac{\partial f}{\partial \psi}=&(\rho_{1}-\rho_{2}) C_hT(1-\ln \frac{T}{T_0})+\rho (C_{h,1}-C_{h,2})T(1-\ln \frac{T}{T_0}),\\
w=&\frac{W^{\prime}(\psi)}{\epsilon} -\epsilon \Delta \psi, \label{continue_w}\\
\mu_{c}=&\mu_{0}+\alpha p,
    \label{continue_mu_c-2}
    \\
\boldsymbol{q}_{E}=&-k\boldsymbol{\nabla} T-\lambda_{u} \epsilon (\boldsymbol{\nabla} \psi\otimes \boldsymbol{\nabla} \psi)\cdot \boldsymbol{v}-\lambda_{u} \epsilon \boldsymbol{\nabla} \psi \frac{\partial \psi}{\partial t}-m_{\psi}\mu_{c}\boldsymbol{\nabla} \mu_{c}.
\end{align}

\subsection{Non-dimensionalization} 
We now non-dimensionalize the model Eqs.~(\ref{model_momentum-conservation}) -- (\ref{model_energy_conservation}). Using the properties of fluid 1 as the characteristic quantities, we non-dimensionalize the variable properties as
\begin{align}
\bar{\Lambda}&=\psi + \zeta_{\Lambda} (1-\psi) ,
\end{align}
where $\Lambda$ stands for $\rho, \mu, c_{hc}$ and $k$, respectively, and $\zeta_{\Lambda}=\Lambda_{2}/\Lambda_{1}$ are the corresponding physical property ratios. 
By selecting $R^*$, $V^*$, and $T^{*}$ as the characteristic length, velocity, and temperature, respectively, and $\mu^*_{c}=\rho_1 (V^*)^2$ as the characteristic chemical potential, energy density, and pressure, $\sigma_{0}$ as the characteristic surface tension, and ${m}^*_{\psi}=m$ as the characteristic mobility for the phase field, we non-dimensionalize the physical quantities as follows:
\begin{align}
&\bar{t}=\frac{tV^*}{R^*},~~\bar{\bs{v}}=\frac{\bs{v}}{V^*},~~\bar{T}=\frac{T}{T^{*}},~~\bar{T}_{0}=\frac{{T}_{0}}{T^{*}},
~~\bar{\epsilon}=\frac{\epsilon}{R^*},~~\bar{\mu}_{c}=\frac{\mu_{c}}{\rho_1 (V^*)^2},\notag\\
&\bar{\hat{f}}=\frac{\hat{f}}{\rho_1 (V^*)^2},~~\bar{\hat{u}}=\frac{\hat{u}}{\rho_1 (V^*)^2},~~\bar{\hat{s}}=\frac{\hat{s}T^{*}}{\rho_1 (V^*)^2},~~\bar{p}=\frac{p}{\rho_1 (V^*)^2},\notag\\
&\bar{\lambda}_{u}=\frac{\lambda_{u}}{\sigma_{0}},~~\bar{\lambda}_{f}=\frac{\lambda_{f}}{\sigma_{0}},~~\bar{\lambda}_{s}=\frac{\lambda_{s}T^{*}}{\sigma_{0}},~~\bar{m}_{\psi}=\frac{{m}_{\psi}}{m}.
\end{align}
After dropping the bar notations, the non-dimensional phase-field model for two-phase flows with thermocapillary effects can be given as follows:
    \begin{align}
\rho \frac{\partial \boldsymbol{v}}{\partial t}+\rho\boldsymbol{v}\cdot\boldsymbol{\nabla} \boldsymbol{v}&=\boldsymbol{\nabla} \cdot \mathbf{T}-\frac{\rho}{{\rm Fr}} \hat{\boldsymbol{z}},
    \label{sys-vel-nodim}
    \\
\boldsymbol{\nabla}\cdot \boldsymbol{v}& = \frac{ \alpha}{\mathrm{Pe}_{\psi}} \boldsymbol{\nabla} \cdot \left(m_{\psi} \boldsymbol{\nabla} {\mu}_{c}\right),
    \label{sys-mass-nodim}
    \\
\frac{\partial \psi}{\partial t}+\boldsymbol{\nabla} \cdot(\psi \boldsymbol{v}) &= \frac{1}{\mathrm{Pe}_{\psi}} \boldsymbol{\nabla} \cdot \left( m_{\psi}\boldsymbol{\nabla} {{\mu}}_{c} \right),
    \label{sys-phase-nodim}
    \\
\frac{\partial \left(\rho C_hT\right)}{\partial t}+\boldsymbol{\nabla}\cdot (\rho C_hT \boldsymbol{v}) &= -\boldsymbol{\nabla} \cdot \boldsymbol{q}_{E}+{\rm Ec}\boldsymbol{\nabla}\boldsymbol{v}:\mathbf{T}
    \nonumber
    \\
& - \frac{{\rm Ec}}{{\rm We}}\lambda_{u}\left(\frac{\partial \delta}{\partial t}+\boldsymbol{\nabla}\cdot (\delta \boldsymbol{v})\right) ,
    \label{sys-energy-nodim}
    \end{align}
where
    \begin{align}
\Lambda &= \psi + \zeta_{\Lambda} (1-\psi), \quad \zeta_{\Lambda}=\Lambda_{2}/\Lambda_{1}~~{\rm for}~~\Lambda=\rho,\mu,C_h,k,
    \label{property-nodim}
    \\
\mathbf{T}&=-p \mathbf{I}+ \hat{f} \mathbf{I}+ \frac{1}{\mathrm{Re}} \mathbf{\tau} -\frac{1}{\rm We}\lambda_{f}(T) \epsilon (\boldsymbol{\nabla} \psi\otimes \boldsymbol{\nabla} \psi)-\mu_{0}\psi \mathbf{I},
    \label{stress-tensor-nodim}
    \\
\hat{f}&=f(T, \psi)+\frac{1}{\rm We}\lambda_{f}(T) \delta, \label{non-dimensional-free-en}\\
f &= \frac{1}{{\rm Ec}} \rho C_hT -\frac{1}{{\rm Ec}} \rho C_hT\ln \left(\frac{T}{T_0}\right),
    \\
\lambda_{f}(T)&= \eta\left(1 - {\rm Ca}~{\rm Ma} (T-T_{0}) \right),
    \\
\delta&=\frac{W(\psi)}{\epsilon} +\epsilon \frac{|\boldsymbol{\nabla} \psi|^2}{2},
    \label{definition_delta_n+1_nodim}
    \\
W(\psi)&=\frac{\psi^2(1-\psi)^2}{4},
    \\
\mathbf{\tau}&=\mathcal{\mu}\left(\boldsymbol{\nabla} \boldsymbol{v}+\boldsymbol{\nabla} \boldsymbol{v}^{\top}\right)-\frac{2}{3} \mathcal{\mu}(\boldsymbol{\nabla} \cdot \boldsymbol{v}) \mathbf{I},
    \\
\mu_{0}&=\frac{\partial f}{\partial \psi}+\frac{1}{\rm We}\lambda_{f}(T) w,
    \label{continue_mu_0-nodim}
    \\
\frac{\partial f}{\partial \psi}&=\frac{1}{{\rm Ec}} \frac{\rho_{1}-\rho_{2}}{\rho_{1}} C_hT\left(1-\ln \left(\frac{T}{T_0}\right)\right)+\frac{1}{{\rm Ec}} \rho \frac{C_{h,1}-C_{h,2}}{C_{h,1}}T\left(1-\ln \left(\frac{T}{T_0}\right)\right),\\
w&=\frac{W^{\prime}(\psi)}{\epsilon} -\epsilon \Delta \psi,
    \label{continue_w-nodim}
    \\
W^{\prime}(\psi)&=\psi(\psi-1)(\psi-\frac{1}{2}),\\
m_{\psi}&=\sqrt{\psi^2(1-\psi)^2},
    \\
\mu_{c}&=\mu_{0}+\alpha p,
    \label{continue_mu_c}
    \\
\boldsymbol{q}_{E}&=-\frac{1}{\mathrm{Pe}_{T}}k\boldsymbol{\nabla} T-\frac{{\rm Ec}}{{\rm We}}\lambda_{u} \epsilon (\boldsymbol{\nabla} \psi\otimes \boldsymbol{\nabla} \psi)\cdot \boldsymbol{v}-\frac{{\rm Ec}}{{\rm We}}\lambda_{u} \epsilon \boldsymbol{\nabla} \psi \frac{\partial \psi}{\partial t}
    \nonumber
    \\
&-\frac{{\rm Ec}}{\mathrm{Pe}_{\psi}} m_{\psi}\mu_{c}\boldsymbol{\nabla} \mu_{c},
    \\
\lambda_{u}&= \eta(1+{\rm Ca}~{\rm Ma} T_{0}).
    \end{align}
 Here the non-dimensional parameters are given by
    \begin{equation}
    \begin{split}
&\mathrm{Re} = \frac{\rho_{1} V^* R^*}{\mu_{1}},~~{\rm Ca}= \frac{\mu_{1}V^*}{\sigma_{0}},~~{\rm Ma} = \frac{\sigma_{T}T^{*}}{\mu_{1} V^*},~~{\rm Fr} = \frac{(V^*)^2}{gR^*},
    \\
&\mathrm{Pe}_{\psi} = \frac{V^* R^*}{m^*_{\psi}\mu^*_{c} },~~\mathrm{Pe}_{T}=\frac{\rho_{1}C_{h,1}V^*R^*}{k_{1}},~~{\rm Ec}=\frac{(V^{*})^2}{C_{h,1}T^{*}},~~{\rm We}={\rm Ca}~\mathrm{Re},
    \end{split}
    \label{non-dim-para}
\end{equation}
where $\mathrm{Re}$ is the Reynolds number, ${\rm Ca}$ is the Capillary number, ${\rm Ma}$ is the thermal Marangoni number, ${\rm Fr}$ is the Froude number, $\mathrm{Pe}_{\psi}$ and $\mathrm{Pe}_{T}$ are the Peclect numbers for the phase-field and heat equations, respectively, ${\rm Ec}$ is the Eckert number and ${\rm We}$ is the Weber number. 
In addition, the non-dimensional internal energy and entropy densities are given by
\begin{align}
\hat{ u}&=u + \frac{1}{{\rm We}}\lambda_{u}\delta,~~~~u=\frac{1}{\rm Ec}\rho C_hT,\label{non-dimensional-internal-en}\\
\hat{s}&= s + \frac{1}{{\rm We}}\lambda_{s}\delta, ~~~~s= \frac{1}{{\rm Ec}} \rho C_h\ln \left(\frac{T}{T_0}\right), ~~~\lambda_{s}= \eta{\rm Ca}~{\rm Ma}.\label{non-dimensional-entropy}
\end{align}

%Note that the initial condition for $\Phi$ and the boundary condition for $T$ may be different for the numerical validation in \S\ref{}.
%\newpage
\section{Model validation}\label{section-3}
In this section, to validate the model derivations, we demonstrate that the conservation laws of mass and energy, and entropy increase can be derived from the model Eqs.~(\ref{sys-vel-nodim})-(\ref{sys-energy-nodim}), which can be further served as the foundation for designing numerical methods.
\begin{thm}\label{them-continous-mass-energy} Consider a closed system in $\Omega$ with following boundary conditions
\begin{align}
\boldsymbol{v}|_{\partial\Omega}=0,~~~\boldsymbol{n}\cdot \boldsymbol{\nabla}\mu_{c}|_{\partial\Omega}=\boldsymbol{n}\cdot\boldsymbol{\nabla}\psi|_{\partial\Omega}=\boldsymbol{n}\cdot\boldsymbol{\nabla} {T}|_{\partial\Omega}= 0, \label{proof_bc-nodim}
\end{align}
the model Eqs.~(\ref{sys-vel-nodim}) -- (\ref{sys-energy-nodim}) satisfy the following mass and energy conservation
\begin{align}
\frac{d}{dt}\int_{\Omega}\rho ~{d}{x}=0~~{\rm and}~~\frac{d E}{d t}=\frac{d}{dt}\int_{\Omega}\left\{\hat{u}+\frac{1}{2} \rho |\boldsymbol{v}|^2+\frac{\rho}{{\rm Fr}} z\right\} d {x}=0,\label{proof_mass-energy-nodim}
\end{align}
where $E$ denotes the dimensionless total energy of the system, comprising the internal energy $\hat{u}$, kinetic energy $\frac{1}{2}\rho|\boldsymbol{v}|^2$, and potential energy $\frac{\rho}{{\rm Fr}} z$.
\end{thm}
%defined in Equation~(\ref{non-dimensional-internal-en})

\begin{proof}
We first show the proof for the mass conservation (\ref{proof_mass-energy-nodim}).
Multiplying Eq.~(\ref{sys-phase-nodim}) by $\alpha$ and using Eq.~(\ref{sys-mass-nodim}), we obtain
\begin{align}
\alpha \psi_{t} +\alpha\boldsymbol{\nabla}\cdot (\psi\boldsymbol{v})= \boldsymbol{\nabla}\cdot \boldsymbol{v}.\label{proof-mass-conservation}
\end{align}
With the help of $\alpha$ (\ref{definition_alpha-nodim}) and $\rho$ (\ref{property-nodim}), Eq.~(\ref{proof-mass-conservation}) can be reformulated as
\begin{align}
\rho_{t} +\boldsymbol{\nabla}\cdot(\rho \boldsymbol{v})= 0.\label{model-mass-con2}
\end{align}
By integrating over $\Omega$ and applying the divergence theorem with the boundary conditions in (\ref{proof_bc-nodim}), we deduce the mass conservation (\ref{proof_mass-energy-nodim}).\\
We now show the proof for the energy conservation.
%in Theorem \ref{them-continous-mass-%energy}.
Multiplying Eq.~(\ref{sys-phase-nodim}) by $-\rho_{2}\alpha$ and using Eq.~(\ref{sys-mass-nodim}), we obtain
\begin{align} \label{rc_energy_1}
-\rho_{2}\alpha\frac{\partial \psi}{\partial t}-\rho_{2}\alpha\boldsymbol{\nabla} \cdot(\psi \boldsymbol{v}) &= -\rho_{2}\boldsymbol{\nabla}\cdot \boldsymbol{v}.
\end{align}
With the definitions of $\alpha$ (\ref{definition_alpha-nodim}) and $\rho$ (\ref{property-nodim}), Eq.~(\ref{rc_energy_1}) can be rewritten as
\begin{align}\label{rc_energy_3}
\frac{\partial \rho}{\partial t}+\boldsymbol{\nabla} \cdot(\rho \boldsymbol{v}) =0.
\end{align}
Multiplying Eq.~(\ref{sys-vel-nodim}) by $\boldsymbol{v}$ and Eq.~(\ref{rc_energy_3}) by $|\boldsymbol{v}|^{2}/2$, and adding them together, we obtain
\begin{align} \label{rc_energy_4}
\frac{1}{2}\frac{\partial ({\rho |\boldsymbol{v}|^{2}})}{\partial {t}} +\frac{1}{2}\boldsymbol{\nabla}\cdot(\rho|\boldsymbol{v}|^{2} \boldsymbol{v})=&\boldsymbol{\nabla} \cdot \mathbf{T}\cdot \boldsymbol{v}-\frac{\rho}{{\rm Fr}} \hat{\boldsymbol{z}}\cdot \boldsymbol{v},
\end{align}
where the following identity is used
\begin{align}
\rho(\boldsymbol{v} \cdot \boldsymbol{\nabla} \boldsymbol{v})\cdot  \boldsymbol{v} + \frac{1}{2}\boldsymbol{\nabla} \cdot(\rho \boldsymbol{v}) |\boldsymbol{v}|^{2}=\frac{1}{2}\boldsymbol{\nabla}\cdot(\rho|\boldsymbol{v}|^{2} \boldsymbol{v}).
\end{align}
In addition, with the help of $\hat{u}$ (\ref{non-dimensional-internal-en}), Eq.~(\ref{sys-energy-nodim}) can be rewritten as
\begin{align} \label{rc_energy_4-1}
\frac{\partial \hat{u}}{\partial t}+\boldsymbol{\nabla} \cdot(\hat{u}\boldsymbol{v})=-\frac{1}{{\rm Ec}}\boldsymbol{\nabla} \cdot \boldsymbol{q}_{E}+\boldsymbol{\nabla} \boldsymbol{v}: \mathbf{T}.
\end{align}
Combining Eqs.~(\ref{rc_energy_4}) and (\ref{rc_energy_4-1}), and integrating over $\Omega$ while applying the divergence theorem with the boundary conditions in Eq.~(\ref{proof_bc-nodim}), we obtain
\begin{align}\label{rc_energy_5}
\frac{d}{dt}\int_{\Omega}\frac{1}{2} \rho |\boldsymbol{v}|^{2} +\hat{u}~dx=\int_{\Omega} -\frac{\rho}{{\rm Fr}} \hat{\boldsymbol{z}}\cdot \boldsymbol{v}~dx,
\end{align}
where the following identity is used
\begin{align}
\boldsymbol{\nabla} \cdot (\mathbf{T}\cdot \boldsymbol{v})=\boldsymbol{\nabla} \cdot \mathbf{T}\cdot \boldsymbol{v}+\boldsymbol{\nabla} \boldsymbol{v}: \mathbf{T}.
\end{align}
In addition, using Eq.~(\ref{rc_energy_3}), we have
\begin{align}\label{rc_energy_7}
\frac{d}{dt}\int_{\Omega}\frac{\rho}{{\rm Fr}} z~dx=\int_{\Omega} \frac{\partial \rho}{\partial t}\frac{z}{{\rm Fr}}~dx
=\int_{\Omega} -\boldsymbol{\nabla}\cdot(
\rho \boldsymbol{v})\frac{z}{{\rm Fr}}~dx.
\end{align}
Adding Eqs.~(\ref{rc_energy_5}) and (\ref{rc_energy_7}) together, and applying the divergence theorem with the boundary conditions in (\ref{proof_bc-nodim}), we deduce the energy conservation (\ref{proof_mass-energy-nodim}).
\end{proof}

\begin{thm}\label{them-continous-entropy} Consider a closed system in $\Omega$ with following boundary conditions
\begin{align}
\boldsymbol{v}|_{\partial\Omega}=0,~~~\boldsymbol{n}\cdot \boldsymbol{\nabla}\mu_{c}|_{\partial\Omega}=\boldsymbol{n}\cdot\boldsymbol{\nabla}\psi|_{\partial\Omega}=\boldsymbol{n}\cdot\boldsymbol{\nabla} {T}|_{\partial\Omega}= 0, \label{rc_bc}
\end{align}
the model Eqs.~(\ref{sys-vel-nodim}) -- (\ref{sys-energy-nodim}) satisfy the following entropy increase
\begin{align}
\frac{d S}{d t}&=\frac{d}{dt}\int_{\Omega}\left\{\frac{1}{{\rm Ec}} \rho C_h\ln \left(\frac{T}{T_0}\right) + \frac{1}{{\rm We}}\lambda_{s}\delta\right\} d {x} \notag \\&=\int_{\Omega}\bigg\{\frac{k}{{\rm Ec}~\mathrm{Pe}_{T}}\frac{|\boldsymbol{\nabla} T|^2}{T^2}+\frac{1}{\mathrm{Re}}\frac{\mathbf{\tau}:\boldsymbol{\nabla} \boldsymbol{v}}{T}+\frac{m_{\psi}}{\mathrm{Pe}_{\psi}}\frac{|\boldsymbol{\nabla}\mu_{c}|^2}{T}\bigg\} d {x} \geq0, \label{proof_entropy-nodim}
\end{align}
where the entropy $S$ of the system is defined as
\begin{align}
	S=\int_{\Omega} \hat{s}(T, \psi, \boldsymbol{\nabla} \psi) d {x}, \label{non-dimension-total-entropy}
\end{align}
and $\hat{s}$ is defined in Eq.~(\ref{non-dimensional-entropy}).
%We refer \ref{proof_continous entropy %increase} for the detailed proof.
\end{thm}

\begin{proof}
Multiplying Eq.~(\ref{sys-phase-nodim}) by $\mu_{0}$, we obtain
\begin{align} \label{rc_entropy_0}
\frac{\partial \psi}{\partial t}  \mu_{0}+\boldsymbol{\nabla} \cdot(\psi \boldsymbol{v})\mu_{0} =&\frac{1}{\mathrm{Pe}_{\psi}}  \boldsymbol{\nabla} \cdot \left(m_{\psi} \boldsymbol{\nabla} {\mu}_{c}\right)\mu_{0}.
\end{align}
Using the definitions of $\mu_{0}$ (Eq.~(\ref{continue_mu_0-nodim})) and $w$ (Eq.~(\ref{continue_w-nodim})), the first term leads to
\begin{align} \label{rc_entropy_1}
\frac{\partial \psi}{\partial t} \mu_{0}=\frac{\partial \psi}{\partial t}\frac{\partial f}{\partial \psi}+\frac{1}{\rm We}\lambda_{f}(T)\frac{\partial \psi}{\partial t}\frac{W^{\prime}(\psi)}{\epsilon}-\frac{1}{\rm We}\lambda_{f}(T)\frac{\partial \psi}{\partial t}\epsilon \Delta \psi.
\end{align}
Furthermore, we note that
\begin{align} \label{rc_entropy_3}
\lambda_{f}(T)\frac{\partial \psi}{\partial t}\frac{W^{\prime}(\psi)}{\epsilon}&=\lambda_{f}(T)\frac{1}{{\epsilon}}\frac{\partial W}{\partial t},
\end{align}
and
\begin{align} \label{rc_entropy_4}
-\lambda_{f}(T)\frac{\partial \psi}{\partial t} \epsilon \Delta \psi
=&-\lambda_{f}(T) \epsilon\boldsymbol{\nabla}\cdot(\frac{\partial \psi}{\partial t}\boldsymbol{\nabla}\psi)+\frac{1}{2}\lambda_{f}(T) \epsilon \frac{\partial |\boldsymbol{\nabla} \psi|^2}{\partial t}.
\end{align}
Substituting Eqs.~(\ref{rc_entropy_3}) and (\ref{rc_entropy_4}) into Eq.~(\ref{rc_entropy_1}), and using the definition of $\delta$ (Eq.~(\ref{definition_delta_n+1_nodim})), the first term of Eq.~(\ref{rc_entropy_0}) can be rewritten as
\begin{align}
\frac{\partial \psi}{\partial t} \mu_{0}
&=\frac{\partial \psi}{\partial t}\frac{\partial f}{\partial \psi}+\frac{1}{\rm We}\lambda_{f}(T)\frac{\partial \delta}{\partial t}
-\frac{1}{\rm We}\lambda_{f}(T)\epsilon\boldsymbol{\nabla}\cdot(\frac{\partial \psi}{\partial t}\boldsymbol{\nabla}\psi). \label{rc_entropy_6}
\end{align}
The second term of Eq.~(\ref{rc_entropy_0}) leads to
\begin{align}\label{rc_entropy_7}
\boldsymbol{\nabla} \cdot(\psi \boldsymbol{v})\mu_{0}&=\mu_{0}\psi (\boldsymbol{\nabla}\cdot \boldsymbol{v})+  \mu_{0}(\boldsymbol{v} \cdot \boldsymbol{\nabla} \psi) \notag\\
&=\mu_{0}\psi (\boldsymbol{\nabla}\cdot \boldsymbol{v})+\frac{\partial f}{\partial \psi}(\boldsymbol{v} \cdot \boldsymbol{\nabla} \psi)+\frac{1}{\rm We}\lambda_{f}(T)(\boldsymbol{v} \cdot \boldsymbol{\nabla} \psi) w,
\end{align}
where Eq.~(\ref{continue_mu_0-nodim}) (the definition of $\mu_{0}$) is used. Furthermore, we obtain
\begin{align}\label{rc_entropy_8}
\lambda_{f}(T) (\boldsymbol{v} \cdot \boldsymbol{\nabla} \psi)w&=\lambda_{f}(T)(\boldsymbol{v} \cdot \boldsymbol{\nabla} \psi)\frac{W^{\prime}(\psi)}{\epsilon} -\lambda_{f}(T) (\boldsymbol{v} \cdot \boldsymbol{\nabla} \psi)\epsilon \Delta \psi \notag\\
&=\lambda_{f}(T)(\boldsymbol{v} \cdot \boldsymbol{\nabla} \delta)-\lambda_{f}(T)\epsilon \boldsymbol{\nabla}\cdot \left((\boldsymbol{\nabla} \psi\otimes \boldsymbol{\nabla} \psi)\cdot \boldsymbol{v}\right) \notag\\
&+\lambda_{f}(T)\epsilon\boldsymbol{\nabla} \boldsymbol{v}:(\boldsymbol{\nabla} \psi\otimes \boldsymbol{\nabla} \psi),
\end{align}
where the following identity and the definition of $\delta$ (Eq.~(\ref{definition_delta_n+1_nodim})) are used
\begin{align}
-(\boldsymbol{v} \cdot \boldsymbol{\nabla} \psi) \Delta \psi&=-\boldsymbol{\nabla}\cdot \left((\boldsymbol{\nabla} \psi\otimes \boldsymbol{\nabla} \psi)\cdot \boldsymbol{v}\right)+\boldsymbol{\nabla} \boldsymbol{v}:(\boldsymbol{\nabla} \psi\otimes \boldsymbol{\nabla} \psi)  \notag\\
&+\boldsymbol{v} \cdot \boldsymbol{\nabla}\frac{|\boldsymbol{\nabla}\psi|^{2}}{2}.
\end{align}
Substituting Eq.~(\ref{rc_entropy_8}) into Eq.~(\ref{rc_entropy_7}), the second term of Eq.~(\ref{rc_entropy_0}) can be rewritten as
\begin{align}\label{rc_entropy_9}
\boldsymbol{\nabla} \cdot(\psi \boldsymbol{v})\mu_{0}
&=\mu_{0}\psi (\boldsymbol{\nabla}\cdot \boldsymbol{v})+\frac{\partial f}{\partial \psi}(\boldsymbol{v} \cdot \boldsymbol{\nabla} \psi) +\frac{1}{\rm We}\lambda_{f}(T)(\boldsymbol{v} \cdot \boldsymbol{\nabla} \delta)\notag\\
&-\frac{1}{\rm We}\lambda_{f}(T)\epsilon \boldsymbol{\nabla}\cdot \left((\boldsymbol{\nabla} \psi\otimes \boldsymbol{\nabla} \psi)\cdot \boldsymbol{v}\right)+\frac{1}{\rm We}\lambda_{f}(T)\epsilon\boldsymbol{\nabla} \boldsymbol{v}:(\boldsymbol{\nabla} \psi\otimes \boldsymbol{\nabla} \psi).
\end{align}
Substituting Eqs.~(\ref{rc_entropy_6}) and (\ref{rc_entropy_9}) into Eq.~(\ref{rc_entropy_0}), we obtain
\begin{align} \label{rc_entropy_10}
\frac{\partial \psi}{\partial t}\frac{\partial f}{\partial \psi}&=-\frac{\partial f}{\partial \psi}(\boldsymbol{v} \cdot \boldsymbol{\nabla} \psi)-\frac{1}{\rm We}\lambda_{f}(T)\frac{\partial \delta}{\partial t} \notag\\
&+\frac{1}{\rm We}\lambda_{f}(T)\epsilon\boldsymbol{\nabla}\cdot \left(\frac{\partial \psi}{\partial t}\boldsymbol{\nabla}\psi+(\boldsymbol{\nabla} \psi\otimes \boldsymbol{\nabla} \psi)\cdot \boldsymbol{v}\right) \notag\\
&-\mu_{0}\psi (\boldsymbol{\nabla}\cdot \boldsymbol{v})-\frac{1}{\rm We}\lambda_{f}(T)(\boldsymbol{v} \cdot \boldsymbol{\nabla} \delta) \notag\\
%&+\lambda_{f}(T^{n})\epsilon \boldsymbol{\nabla}\cdot \left((\boldsymbol{\nabla} \psi^{n+1}\otimes \boldsymbol{\nabla} \psi^{n+1})\cdot \boldsymbol{v}^{n}\right)\\ \notag
&-\frac{1}{\rm We}\lambda_{f}(T)\epsilon\boldsymbol{\nabla} \boldsymbol{v}:(\boldsymbol{\nabla} \psi\otimes \boldsymbol{\nabla} \psi)+\frac{1}{\mathrm{Pe}_{\psi}}  \boldsymbol{\nabla} \cdot \left(m_{\psi} \boldsymbol{\nabla} {\mu}_{c}\right)\mu_{0}.
\end{align}
Multiplying Eq.~(\ref{sys-mass-nodim}) by $p$, and adding it to Eq.~(\ref{rc_entropy_10}), we obtain
\begin{align} \label{rc_entropy_12}
\frac{\partial \psi}{\partial t}\frac{\partial f}{\partial \psi}&=-\frac{\partial f}{\partial \psi}(\boldsymbol{v} \cdot \boldsymbol{\nabla} \psi)-\frac{1}{\rm We}\lambda_{f}(T)\frac{\partial \delta}{\partial t}  \notag\\
&+\frac{1}{\rm We}\lambda_{f}(T)\epsilon\boldsymbol{\nabla}\cdot\left(\frac{\partial \psi}{\partial t}\boldsymbol{\nabla}\psi+(\boldsymbol{\nabla} \psi\otimes \boldsymbol{\nabla} \psi)\cdot \boldsymbol{v}\right) \notag\\
&-\frac{1}{\rm We}\lambda_{f}(T)(\boldsymbol{v} \cdot \boldsymbol{\nabla} \delta) \notag\\
&-\frac{1}{\rm We}\lambda_{f}(T)\epsilon\boldsymbol{\nabla} \boldsymbol{v}:(\boldsymbol{\nabla} \psi\otimes \boldsymbol{\nabla} \psi) \notag\\
&-\mu_{0}\psi (\boldsymbol{\nabla}\cdot \boldsymbol{v})-p(\boldsymbol{\nabla} \cdot \boldsymbol{v}) \notag\\
&+\frac{1}{\mathrm{Pe}_{\psi}}\boldsymbol{\nabla} \cdot (m_{\psi}\mu_{c} \boldsymbol{\nabla}  \mu_{c} ) -\frac{m_{\psi}}{\mathrm{Pe}_{\psi}}|\boldsymbol{\nabla}\mu_{c}|^2,
\end{align}
where the following  identity is used
\begin{align}
&\frac{1}{\mathrm{Pe}_{\psi}}  \boldsymbol{\nabla} \cdot \left(m_{\psi} \boldsymbol{\nabla} {\mu}_{c}\right)\mu_{0}+\frac{ \alpha}{\mathrm{Pe}_{\psi}} \boldsymbol{\nabla} \cdot \left(m_{\psi} \boldsymbol{\nabla} {\mu}_{c}\right) p \notag\\
=&\frac{1}{\mathrm{Pe}_{\psi}}  \boldsymbol{\nabla} \cdot \left(m_{\psi} \boldsymbol{\nabla} {\mu}_{c}\right)\mu_{c}\notag\\
=&\frac{1}{\mathrm{Pe}_{\psi}}\boldsymbol{\nabla} \cdot (m_{\psi}\mu_{c}\boldsymbol{\nabla}  \mu_{c} ) -\frac{m_{\psi}}{\mathrm{Pe}_{\psi}}|\boldsymbol{\nabla}\mu_{c}|^2.\notag
\end{align}
In addition, with the help of the definition of $u$ (\ref{non-dimensional-internal-en}), Eq.~(\ref{sys-energy-nodim}) can be rewritten as
\begin{align} \label{rc_entropy_12-1}
\frac{\partial u}{\partial t}+\boldsymbol{\nabla} \cdot(u\boldsymbol{v})=-\frac{1}{{\rm Ec}}\boldsymbol{\nabla} \cdot \boldsymbol{q}_{E}+\boldsymbol{\nabla} \boldsymbol{v}: \mathbf{T}- \frac{1}{{\rm We}}\lambda_{u}\left(\frac{\partial \delta}{\partial t}+\boldsymbol{\nabla}\cdot (\delta \boldsymbol{v})\right).
\end{align}
Subtracting Eq.~(\ref{rc_entropy_12}) from Eq.~(\ref{rc_entropy_12-1}), we obtain
\begin{align}  \label{rc_entropy_13-1}
\frac{\partial u}{\partial t}-\frac{\partial f}{\partial \psi}\frac{\partial \psi}{\partial t}&=-\boldsymbol{v}\cdot \boldsymbol{\nabla} u +\frac{\partial f}{\partial \psi}(\boldsymbol{v} \cdot \boldsymbol{\nabla} \psi)
-\frac{1}{{\rm We}}T\lambda_{s}\frac{\partial \delta}{\partial t} \notag\\
&-\frac{1}{{\rm We}}T\lambda_{s}(\boldsymbol{v}\cdot \boldsymbol{\nabla} \delta) +  \frac{1}{{\rm Ec}~\mathrm{Pe}_{T}}\boldsymbol{\nabla} \cdot (k\boldsymbol{\nabla} T) \notag\\
&+\frac{1}{\rm We}T\lambda_{s}\epsilon \boldsymbol{\nabla} \cdot \left(\boldsymbol{\nabla}\psi  \frac{\partial \psi}{\partial t} +\boldsymbol{\nabla}\psi (\boldsymbol{v}\cdot \boldsymbol{\nabla} \psi) \right) \nonumber
    \\
&-T\hat{s}(\boldsymbol{\nabla} \cdot \boldsymbol{v})+ \frac{1}{\mathrm{Re}} \mathbf{\tau}:\boldsymbol{\nabla} \boldsymbol{v}+\frac{m_{\psi}}{\mathrm{Pe}_{\psi}}|\boldsymbol{\nabla}\mu_{c}|^2,
\end{align}
where the following identities are used
\begin{align}
\frac{1}{{\rm We}}\lambda_{u}\frac{\partial \delta}{\partial t}-\frac{1}{\rm We}\lambda_{f}(T)\frac{\partial \delta}{\partial t}&=\frac{1}{\rm We}T\lambda_{s}\frac{\partial \delta}{\partial t},\notag\\
\boldsymbol{\nabla} \cdot(u\boldsymbol{v})+ \frac{1}{{\rm We}}\lambda_{u}\boldsymbol{\nabla}\cdot (\delta \boldsymbol{v})-\boldsymbol{\nabla}\boldsymbol{v}: \hat{f}\mathbf{I}&=
\boldsymbol{v}\cdot \boldsymbol{\nabla} u +T\hat{s} (\boldsymbol{\nabla}\cdot\boldsymbol{v})+\frac{1}{{\rm We}}\lambda_{u} (\boldsymbol{v}\cdot\boldsymbol{\nabla} \delta),\notag\\
\frac{1}{{\rm We}}\lambda_{u}(\boldsymbol{v}\cdot \boldsymbol{\nabla}\delta)-\frac{1}{\rm We}\lambda_{f}(T)(\boldsymbol{v}\cdot \boldsymbol{\nabla}\delta)&=\frac{1}{\rm We}T\lambda_{s}(\boldsymbol{v}\cdot \boldsymbol{\nabla}\delta).\notag
\end{align}
Furthermore, Eq.~(\ref{rc_entropy_13-1}) can be rewritten as
    \begin{align}
 \frac{\partial \hat{s}}{\partial t}
&=-\boldsymbol{\nabla}\cdot(\hat{s}\boldsymbol{v}) +\frac{1}{{\rm Ec}~\mathrm{Pe}_{T}}\boldsymbol{\nabla} \cdot (\frac{1}{T}k\boldsymbol{\nabla} T)+ \frac{k}{{\rm Ec}~\mathrm{Pe}_{T}}\frac{|\boldsymbol{\nabla} T|^2}{T^2}
    \nonumber
    \\
&\quad +\frac{1}{\rm We}\lambda_{s}\epsilon \boldsymbol{\nabla} \cdot \left(\boldsymbol{\nabla}\psi  \frac{\partial \psi}{\partial t} +\boldsymbol{\nabla}\psi (\boldsymbol{v}\cdot \boldsymbol{\nabla} \psi) \right)
    \nonumber
    \\
&\quad +\frac{1}{\mathrm{Re}}\frac{\mathbf{\tau}:\boldsymbol{\nabla} \boldsymbol{v}}{T}+\frac{m_{\psi}}{\mathrm{Pe}_{\psi}}\frac{|\boldsymbol{\nabla}\mu_{c}|^2}{T},
     \label{rc_entropy_16}
    \end{align}
with the help of the following identities
\begin{align}
\frac{\partial u}{\partial t}-\frac{\partial \psi}{\partial t}\frac{\partial f}{\partial \psi}&=T\frac{\partial {s}}{\partial t},\notag\\
T\frac{\partial {s}}{\partial t}+\frac{1}{{\rm We}}T\lambda_{s}\frac{\partial \delta}{\partial t}&=T\frac{\partial \hat{s}}{\partial t},\notag\\
\boldsymbol{v}\cdot \boldsymbol{\nabla}u-\frac{\partial f}{\partial \psi}(\boldsymbol{v} \cdot \boldsymbol{\nabla} \psi)&=T\boldsymbol{v}\cdot \boldsymbol{\nabla}s,\notag\\
T\boldsymbol{v}\cdot \boldsymbol{\nabla}s+\frac{1}{{\rm We}}T\lambda_{s}(\boldsymbol{v}\cdot \boldsymbol{\nabla} \delta)+T\hat{s}(\boldsymbol{\nabla} \cdot \boldsymbol{v})&=T\boldsymbol{\nabla}\cdot(\hat{s}\boldsymbol{v}).\notag
\end{align}
Taking the integral over $\Omega$  and applying the divergence theorem with the boundary conditions in Eq.~(\ref{rc_bc}), we finally deduce the entropy increasing Eq.~(\ref{proof_entropy-nodim}).
\end{proof}

%\newpage{rc_entropy_17}

\section{Numerical method} \label{section-discrete-numerical-method}
We now present a first order, semi-decoupled, mass conserving and entropy increasing temporal discrete numerical method for solving the model Eqs.~(\ref{sys-vel-nodim}) -- (\ref{sys-energy-nodim}): 
\begin{align}
\rho^{n} \boldsymbol{v}_{\bar{t}}+\rho^{n} \boldsymbol{v}^{n}\cdot\boldsymbol{\nabla} \boldsymbol{v}^{n+1}=&\boldsymbol{\nabla} \cdot \mathbf{T}^{n+1}-\frac{\rho^{n}}{{\rm Fr}}\hat{\boldsymbol{z}},\label{discrete_momentum-conservation}\\
\boldsymbol{\nabla} \cdot \boldsymbol{v}^{n}=&\frac{ \alpha}{\mathrm{Pe}_{\psi}} \boldsymbol{\nabla} \cdot \left(m_{\psi}^{n} \boldsymbol{\nabla} {\mu}_{c}^{n+1}\right),\label{discrete_mass-conservation}\\
\psi_{\bar{t}}+\boldsymbol{\nabla} \cdot(\psi^{n+1} \boldsymbol{v}^{n}) =&\frac{1}{\mathrm{Pe}_{\psi}}  \boldsymbol{\nabla} \cdot \left(m_{\psi}^{n} \boldsymbol{\nabla} {\mu}_{c}^{n+1}\right), \label{discrete_volume_conservation}\\
%\mu_{0}^{n+1}=&\frac{\partial f}{\partial \psi}+\lambda_{f}(T^{n}) w^{n+1},\\
(\rho C_h)^{n+1} T_{\bar{t}}+(\rho C_h)^{n+1} (\boldsymbol{v}^{n}\cdot \boldsymbol{\nabla} T^{n})=&-\boldsymbol{\nabla} \cdot \boldsymbol{q}_{E}^{n+1}+hs, \label{discrete_energy_conservation}
\end{align}
where
    \begin{align}
\Lambda^{n+1}&=\psi^{n+1}+\zeta_{\Lambda}(1-\psi^{n+1}), \quad \zeta_{\Lambda}=\Lambda_{2}/\Lambda_{1}~~\mbox{for}~~\Lambda=\rho,\mu,C_h,k,
    \label{discrete_property}
    \\
\mathbf{T}^{n+1}&=\mathbf{m}^{n+1}+\frac{1}{\mathrm{Re}} \mathbf{\tau}^{n+1}+ \hat{f}^{n+1} \mathbf{I},
    \\
\mathbf{m}^{n+1}&=-p^{n+1} \mathbf{I} -\frac{1}{\rm We}\lambda_{f}(T^{n}) \epsilon (\boldsymbol{\nabla} \psi^{n+1}\otimes \boldsymbol{\nabla} \psi^{n+1})-\mu_{0}^{n+1}\psi^{n+1}\mathbf{I},
    \\
\mu_{0}^{n+1}&=\frac{\partial f}{\partial \psi}+\frac{1}{\rm We}\lambda_{f}(T^{n}) w^{n+1},
    \label{mu_0_n+1}
    \\
\frac{\partial f}{\partial \psi}& =\frac{1}{{\rm Ec}} \frac{\rho_{1}-\rho_{2}}{\rho_{1}} C_h^{n+1}T^{n}\left(1-\ln \left(\frac{T^{n}}{T_0}\right)\right)
    \nonumber
    \\
&+\frac{1}{{\rm Ec}} \rho^{n} \frac{C_{h,1}-C_{h,2}}{C_{h,1}}T^{n}\left(1-\ln\left( \frac{T^{n}}{T_0}\right)\right),
    \label{discrete_f_psi}
    \\
w^{n+1}&=\frac{W^{\prime}(\psi^{n+1})}{\epsilon} -\epsilon \Delta \psi^{n+1},
    \label{definition_w_n+1}
    \\
W^{\prime}(\psi^{n+1})&=\psi^{n+1}(\psi^{n+1}-1)(\psi^{n+1}-\frac{1}{2}),
    \\
\mathbf{\tau}^{n+1}&=\mu^{n}\left(\boldsymbol{\nabla} \boldsymbol{v}^{n+1}+(\boldsymbol{\nabla} {\boldsymbol{v}^{n+1}})^{\top}\right)-\frac{2}{3} \mu^{n}(\boldsymbol{\nabla} \cdot \boldsymbol{v}^{n+1}) \mathbf{I},
    \\
\hat{f}^{n+1}&=f^{n+1}+\frac{1}{\rm We}\lambda_{f}(T^{n}) \delta^{n+1},
    \\
f^{n+1} &=\frac{1}{{\rm Ec}}\rho^{n+1} C_h^{n+1}T^{n}\left(1-\ln\left( \frac{T^{n}}{T_0}\right)\right),
    \\
\delta^{n+1}&=\frac{W(\psi^{n+1})}{\epsilon} +\epsilon \frac{|\boldsymbol{\nabla} \psi^{n+1}|^2}{2},
    \label{definition_delta_n+1-2}
    \\
W(\psi^{n+1})&=\frac{(\psi^{n+1})^2(1-\psi^{n+1})^2}{4},
    \\
\mu_{c}^{n+1}&=\mu_{0}^{n+1}+\alpha p^{n+1},
    \label{mu_c}
    \\
\boldsymbol{q}_{E}^{n+1}&=-\frac{1}{\mathrm{Pe}_{T}}k\frac{\boldsymbol{\nabla} T^{n+1}+\boldsymbol{\nabla} T^{n}}{2}-\frac{{\rm Ec}}{\mathrm{Pe}_{\psi}} m_{\psi}^{n}\mu_{c}^{n+1}\boldsymbol{\nabla} \mu_{c}^{n+1}
    \nonumber
    \\
&-\frac{{\rm Ec}}{{\rm We}} \lambda_{u} \epsilon \left(\boldsymbol{\nabla}\psi^{n+1} \psi_{\bar{t}}+(\boldsymbol{\nabla}\psi^{n+1} \otimes \boldsymbol{\nabla} \psi^{n+1}) \cdot \boldsymbol{v}^{n} \right),\\
hs&=-{\rm Ec}\frac{\partial u}{\partial \psi} \psi_{\bar{t}}-{\rm Ec}\frac{\partial \tilde{u}}{\partial \psi}(\boldsymbol{v}^{n}\cdot \boldsymbol{\nabla} \psi^{n+1})+{\rm Ec}~\mathbf{m}^{n+1}:\boldsymbol{\nabla} \boldsymbol{v}^{n}
    \nonumber
    \\
& \quad +\frac{{\rm Ec}}{\mathrm{Re}}\mathbf{\tau}^{n+1}:\boldsymbol{\nabla} \boldsymbol{v}^{n+1}-{\rm Ec}~T^{n}\tilde{s}^{n+1}(\boldsymbol{\nabla} \cdot \boldsymbol{v}^{n})-\frac{{\rm Ec}}{\rm We}\lambda_{u}\delta_{\bar{t}}
    \nonumber
    \\
&\quad -\frac{{\rm Ec}}{\rm We} \lambda_{u}(\boldsymbol{v}^{n}\cdot \boldsymbol{\nabla} \delta^{n+1})+corr^{n},
    \\
\frac{\partial u}{\partial \psi} &= \frac{1}{{\rm Ec}} \frac{\rho_{1}-\rho_{2}}{\rho_{1}} C_h^{n+1} T^{n} +\frac{1}{{\rm Ec}} \rho^{n} \frac{C_{h,1}-C_{h,2}}{C_{h,1}} T^{n},
    \label{discrete_u_psi}
    \\
\frac{\partial \tilde{u}}{\partial \psi}&=\frac{\partial f}{\partial \psi}+T^{n}\frac{\partial \tilde{S}^{n+1}}{\partial \psi^{n+1}},
    \label{discrete_u_tile}
    \\
\tilde{S}^{n+1}&=\frac{1}{{\rm Ec}}\rho^{n+1} C_h^{n+1}\ln \left(\frac{T^{n}}{T_0}\right),
    \\
\tilde{s}^{n+1}&=\tilde{S}^{n+1}+\frac{1}{\rm We}\lambda_{s}\delta^{n+1},
    \label{tilde_s}
    \\
corr^{n}&=corr_{1}+corr_{2}+corr_{3}+corr_{4},
    \label{corr}
    \\
corr_{1}&=\frac{{\rm Ec}}{\rm We}\lambda_{f}(T^{n})\frac{({\psi}^{n+1}-\psi^{n})^2}{8\delta t\epsilon},
    \\
corr_{2}&=-\frac{{\rm Ec}}{\rm We}\lambda_{f}(T^{n})\epsilon\frac{|\boldsymbol{\nabla}\psi^{n+1}
-\boldsymbol{\nabla}{\psi}^{n}|^2}{2\delta t},
    \\
corr_{3}&=\frac{(\rho C_h)^{n+1}T^{n} (T^{n+1}-T^{n})^{2}}{2(T_{min})^{2} \delta t},
    \\
corr_{4}&=\frac{k|\boldsymbol{\nabla} T^{n+1}|^{2}-k|\boldsymbol{\nabla}T^{n}|^2}{4~\mathrm{Pe}_{T}~T^{n}}.
    \end{align}
Here, $\tilde{S}^{n+1}$ is the approximation of $s^{n+1}$, and $\tilde{s}^{n+1}$ is the approximation of ${\hat{s}}^{n+1}$. $\boldsymbol{v}^{n}$, $p^{n}$, $\psi^n$, $\mu_{0}^n$ and $T^n$ are the approximations of $\boldsymbol{v}(n\delta t)$, $p(n\delta t)$, $\psi(n\delta t)$, $\mu_{0}(n\delta t)$ and $T(n\delta t)$ at time $t=n\delta t$, respectively, and $\delta t$ is the time step.
$*_{\bar{t}}=(*^{n+1}-*^{n})/\delta t$, and $T_{min}$ is the global minimum value of the initial temperature.

Next, we prove that the numerical method (\ref{discrete_momentum-conservation}) -- (\ref{discrete_energy_conservation}) preserves the temporal discrete mass conservation and entropy increasing.
\begin{thm}\label{discrete-conservation-law} Consider a closed system in $\Omega$ with following boundary conditions
\begin{align}
\boldsymbol{v}^{n}|_{\partial\Omega}=0,~~~\boldsymbol{n}\cdot \boldsymbol{\nabla}\mu_{c}^{n}|_{\partial\Omega}=\boldsymbol{n}\cdot\boldsymbol{\nabla}\psi^{n}|_{\partial\Omega}=\boldsymbol{n}\cdot\boldsymbol{\nabla} {T^{n}}|_{\partial\Omega}= 0, \label{proof_bc_discrete}
\end{align}
the numerical method (\ref{discrete_momentum-conservation}) -- (\ref{discrete_energy_conservation}) preserves the following temporal discrete mass conservation
\begin{align}
\int_{\Omega}\rho_{\bar{t}}~{d}{x}&=0, \label{dis-mass-law}
\end{align}
and discrete entropy increase corresponding to the continuous counterpart (\ref{proof_entropy-nodim})
\begin{align}
S_{\bar{t}}&=\int_{\Omega}\frac{\hat{s}^{n+1}-\hat{s}^{n}}{\delta t}d {x}=\int_{\Omega}\bigg\{\frac{k}{{\rm Ec}~\mathrm{Pe}_{T}}\frac{|\boldsymbol{\nabla}  T^{n+1}+\boldsymbol{\nabla}T^{n}|^2}{4(T^{n})^2}\notag+\frac{1}{\mathrm{Re}}\frac{\mathbf{\tau}^{n+1}:\boldsymbol{\nabla} \boldsymbol{v}^{n+1}}{T^{n}}\notag\\&+\frac{m_{\psi}^{n}}{\mathrm{Pe}_{\psi}}\frac{|\boldsymbol{\nabla}\mu_{c}^{n+1}|^2}{T^{n}}+res\bigg\} d {x} \geq0,\label{dis-entropy-law}
\end{align}
where $\hat{s}^{n+1}=\frac{1}{{\rm Ec}} \rho^{n+1} C_h^{n+1}\ln \left(\frac{T^{n+1}}{T_0}\right) + \frac{1}{{\rm We}}\lambda_{s}\delta^{n+1}$, and  $res$ is defined in Eq.~(\ref{res}).
%We refer \ref{proof discrete entropy} for the detailed proof.
\end{thm}

\begin{proof}
We first show the proof for the discrete mass conservation.
Multiplying Eq.~(\ref{discrete_volume_conservation}) by $\alpha$ and using Eq.~(\ref{discrete_mass-conservation}), we obtain
\begin{align}
&\alpha\psi_{\bar{t}}+\alpha\boldsymbol{\nabla}\cdot({\psi^{n+1}}{\boldsymbol{v}^{n}}) =\boldsymbol{\nabla}\cdot {\boldsymbol{v}}^{n}. \label{proof-discrete-mass-conservation}
\end{align}
With the help of the definitions of $\alpha$ (\ref{definition_alpha-nodim}) and $\rho^{n+1}$ (\ref{discrete_property}), Eq.~(\ref{proof-discrete-mass-conservation}) can be rewritten as
\begin{align}
&\rho_{\bar{t}}+\boldsymbol{\nabla}\cdot({\rho}^{n+1} \boldsymbol{v}^{n}) = 0.\label{dis-mass-con-1st}
\end{align}
Taking the integral over $\Omega$ with the boundary conditions in (\ref{proof_bc_discrete}), we obtain the discrete mass conservation (\ref{dis-mass-law}).

We now show the proof for the discrete entropy increasing law.
Multiplying Eq.~\eqref{discrete_volume_conservation} by $\mu_{0}^{n+1}$, we obtain
\begin{align} \label{discrete_entropy_0}
\psi_{\bar{t}}  \mu_{0}^{n+1}+\boldsymbol{\nabla} \cdot(\psi^{n+1} \boldsymbol{v}^{n})\mu_{0}^{n+1} =&\frac{1}{\mathrm{Pe}_{\psi}}  \boldsymbol{\nabla} \cdot \left(m_{\psi}^{n} \boldsymbol{\nabla} {\mu}_{c}^{n+1}\right)\mu_{0}^{n+1}.
\end{align}
Using Eq.~(\ref{mu_0_n+1}) (the definition of $\mu_{0}^{n+1}$), the first term leads to
\begin{align} \label{discrete_entropy_1}
\psi_{\bar{t}} \mu_{0}^{n+1}&=\psi_{\bar{t}}\frac{\partial f}{\partial \psi}+\frac{1}{\rm We}\lambda_{f}(T^{n})\psi_{\bar{t}}\frac{W^{\prime}(\psi^{n+1})}{\epsilon}\notag\\
&-\frac{1}{\rm We}\lambda_{f}(T^{n})\psi_{\bar{t}}\epsilon \Delta \psi^{n+1}.
\end{align}
Furthermore, we have
\begin{align} \label{discrete_entropy_3}
\frac{1}{\rm We}\lambda_{f}(T^{n})\psi_{\bar{t}}\frac{W^{\prime}(\psi^{n+1})}{\epsilon}&=\frac{1}{\rm We}\lambda_{f}(T^{n})\frac{W_{\bar{t}}}{\epsilon}+err_1,\\
err_1 &=\frac{1}{\rm We}\lambda_{f}(T^{n})\frac{W''(\xi)
({\psi}^{n}-\psi^{n+1})^2}{2 \delta t \epsilon}, \label{err_1}
%W''(\xi)&=3\xi^{2}-3\xi+\frac{1}{2}, \label{w_prime_prime}
\end{align}
where $\xi$ is between $\psi^{n+1}$ and $\psi^{n}$,
and we have used the following Taylor expansion
\begin{align}
W'(\psi^{n+1})(\psi^{n+1}-{\psi}^{n})&=W(\psi^{n+1})-W({\psi}^{n})+\frac{1}{2}W''(\xi)({\psi}^{n}-\psi^{n+1})^2.
\end{align}
In addition, we have
\begin{align}
-\frac{1}{\rm We}\lambda_{f}(T^{n})\psi_{\bar{t}} \epsilon \Delta \psi^{n+1}&=-\frac{1}{\rm We}\lambda_{f}(T^{n}) \epsilon \boldsymbol{\nabla}\cdot(\frac{\psi^{n+1}-{\psi}^{n}}{\delta t}\boldsymbol{\nabla}\psi^{n+1}) \notag\\
&+\frac{1}{\rm We}\lambda_{f}(T^{n}) \epsilon \left(\frac{|\boldsymbol{\nabla}\psi^{n+1}|^2}{\delta t}-\frac{\boldsymbol{\nabla}
{\psi}^{n}\cdot\boldsymbol{\nabla}\psi^{n+1}}{\delta t} \right) \notag \\
&=-\frac{1}{\rm We}\lambda_{f}(T^{n}) \epsilon\boldsymbol{\nabla}\cdot(\psi_{\bar{t}}\boldsymbol{\nabla}\psi^{n+1})\notag \\
&+\frac{1}{\rm We}\lambda_{f}(T^{n}) \epsilon\frac{|\boldsymbol{\nabla}\psi|^{2}_{\bar{t}}}{2} +err_2, \label{discrete_entropy_4}\\
err_2&=\frac{1}{\rm We}\lambda_{f}(T^{n})\epsilon\frac{|\boldsymbol{\nabla}\psi^{n+1}-\boldsymbol{\nabla}{\psi}^{n}|^2}{2\delta t}. \label{err_2}
\end{align}
%where
%\begin{align}
%\psi_{\bar{t}} = \frac{\psi^{n+1}-\psi^{n}}{\delta t}.
%\end{align}
Substituting Eqs.~(\ref{discrete_entropy_3}) and (\ref{discrete_entropy_4}) into Eq.~(\ref{discrete_entropy_1}), 
 the first term of Eq.~(\ref{discrete_entropy_0}) can be rewritten as
\begin{align}\label{discrete_entropy_6}
\psi_{\bar{t}} \mu_{0}^{n+1}
&=\psi_{\bar{t}}\frac{\partial f}{\partial \psi}+\frac{1}{\rm We}\lambda_{f}(T^{n})\delta_{\bar{t}}
-\frac{1}{\rm We}\lambda_{f}(T^{n})\epsilon\boldsymbol{\nabla}\cdot(\psi_{\bar{t}}\boldsymbol{\nabla}\psi^{n+1}) \notag\\
&+err_1 +err_2.
\end{align}
The second term of Eq.~(\ref{discrete_entropy_0}) leads to
\begin{align}\label{discrete_entropy_7}
\boldsymbol{\nabla} \cdot(\psi^{n+1} \boldsymbol{v}^{n})\mu_{0}^{n+1}&=\mu_{0}^{n+1}\psi^{n+1} (\boldsymbol{\nabla}\cdot \boldsymbol{v}^{n})+\frac{\partial f}{\partial \psi}(\boldsymbol{v}^{n} \cdot \boldsymbol{\nabla} \psi^{n+1}) \notag\\
&+\frac{1}{\rm We}\lambda_{f}(T^{n})(\boldsymbol{v}^{n} \cdot \boldsymbol{\nabla} \psi^{n+1}) w^{n+1},
\end{align}
where the definition of $\mu_{0}^{n+1}$ (Eq.~(\ref{mu_0_n+1})) is used. Furthermore, we have
\begin{align}\label{discrete_entropy_9}
\lambda_{f}(T^{n}) (\boldsymbol{v}^{n} \cdot \boldsymbol{\nabla} \psi^{n+1})w^{n+1}&=\lambda_{f}(T^{n})(\boldsymbol{v}^{n} \cdot \boldsymbol{\nabla} \psi^{n+1})\frac{W^{\prime}(\psi^{n+1})}{\epsilon} \notag\\
&-\lambda_{f}(T^{n}) (\boldsymbol{v}^{n} \cdot \boldsymbol{\nabla} \psi^{n+1})\epsilon \Delta \psi^{n+1} \notag\\
&=\lambda_{f}(T^{n})(\boldsymbol{v}^{n} \cdot \boldsymbol{\nabla} \delta^{n+1}) \notag\\
&-\lambda_{f}(T^{n})\epsilon \boldsymbol{\nabla}\cdot \left((\boldsymbol{\nabla} \psi^{n+1}\otimes \boldsymbol{\nabla} \psi^{n+1})\cdot \boldsymbol{v}^{n}\right) \notag\\
&+\lambda_{f}(T^{n})\epsilon\boldsymbol{\nabla} \boldsymbol{v}^{n}:(\boldsymbol{\nabla} \psi^{n+1}\otimes \boldsymbol{\nabla} \psi^{n+1}),
\end{align}
where the following identity and the definition of $\delta^{n+1}$ (Eq.~(\ref{definition_delta_n+1-2})) are used
\begin{align}
-(\boldsymbol{v}^{n} \cdot \boldsymbol{\nabla} \psi^{n+1}) \Delta \psi^{n+1}&=-\boldsymbol{\nabla}\cdot \left((\boldsymbol{\nabla} \psi^{n+1}\otimes \boldsymbol{\nabla} \psi^{n+1})\cdot \boldsymbol{v}^{n}\right) \notag\\
&+\boldsymbol{\nabla} \boldsymbol{v}^{n}:(\boldsymbol{\nabla} \psi^{n+1}\otimes \boldsymbol{\nabla} \psi^{n+1})+\boldsymbol{v}^{n} \cdot \boldsymbol{\nabla}\frac{|\boldsymbol{\nabla}\psi^{n+1}|^{2}}{2}.
\end{align}
Substituting Eq.~(\ref{discrete_entropy_9}) into Eq.~(\ref{discrete_entropy_7}), the second term of Eq.~(\ref{discrete_entropy_0}) can be rewritten as
\begin{align}\label{discrete_entropy_10}
\boldsymbol{\nabla} \cdot(\psi^{n+1} \boldsymbol{v}^{n})\mu_{0}^{n+1}
&=\mu_{0}^{n+1}\psi^{n+1} (\boldsymbol{\nabla}\cdot \boldsymbol{v}^{n})+\frac{\partial f}{\partial \psi}(\boldsymbol{v}^{n} \cdot \boldsymbol{\nabla} \psi^{n+1}) \notag\\
&+\frac{1}{\rm We}\lambda_{f}(T^{n})(\boldsymbol{v}^{n} \cdot \boldsymbol{\nabla} \delta^{n+1}) \notag\\
&-\frac{1}{\rm We}\lambda_{f}(T^{n})\epsilon \boldsymbol{\nabla}\cdot \left((\boldsymbol{\nabla} \psi^{n+1}\otimes \boldsymbol{\nabla} \psi^{n+1})\cdot \boldsymbol{v}^{n}\right) \notag\\
&+\frac{1}{\rm We}\lambda_{f}(T^{n})\epsilon\boldsymbol{\nabla} \boldsymbol{v}^{n}:(\boldsymbol{\nabla} \psi^{n+1}\otimes \boldsymbol{\nabla} \psi^{n+1}).
\end{align}
fSubstituting Eqs.~(\ref{discrete_entropy_6}) and (\ref{discrete_entropy_10}) into Eq.~(\ref{discrete_entropy_0}), we obtain
\begin{align} \label{discrete_entropy_12}
\psi_{\bar{t}}\frac{\partial f}{\partial \psi}&=-\frac{\partial f}{\partial \psi}(\boldsymbol{v}^{n} \cdot \boldsymbol{\nabla} \psi^{n+1})-\frac{1}{\rm We}\lambda_{f}(T^{n})\delta_{\bar{t}} \notag\\
&+\frac{1}{\rm We}\lambda_{f}(T^{n})\epsilon\boldsymbol{\nabla}\cdot \left(\psi_{\bar{t}}\boldsymbol{\nabla}\psi^{n+1}+(\boldsymbol{\nabla} \psi^{n+1}\otimes \boldsymbol{\nabla} \psi^{n+1})\cdot \boldsymbol{v}^{n}\right) \notag\\
&-\mu_{0}^{n+1}\psi^{n+1} (\boldsymbol{\nabla}\cdot \boldsymbol{v}^{n})-\frac{1}{\rm We}\lambda_{f}(T^{n})(\boldsymbol{v}^{n} \cdot \boldsymbol{\nabla} \delta^{n+1}) \notag\\
&-\frac{1}{\rm We}\lambda_{f}(T^{n})\epsilon\boldsymbol{\nabla} \boldsymbol{v}^{n}:(\boldsymbol{\nabla} \psi^{n+1}\otimes \boldsymbol{\nabla} \psi^{n+1}) \notag\\
&+\frac{1}{\mathrm{Pe}_{\psi}}  \boldsymbol{\nabla} \cdot \left(m_{\psi}^{n} \boldsymbol{\nabla} {\mu}_{c}^{n+1}\right)\mu_{0}^{n+1}-err_1 -err_2.
\end{align}
Multiplying Eq.~(\ref{discrete_mass-conservation}) by $p^{n+1}$, we obtain
\begin{align} \label{discrete_entropy_13}
p^{n+1}(\boldsymbol{\nabla} \cdot \boldsymbol{v}^{n}) &=\frac{ \alpha}{\mathrm{Pe}_{\psi}} \boldsymbol{\nabla} \cdot \left(m_{\psi}^{n} \boldsymbol{\nabla} {\mu}_{c}^{n+1}\right) p^{n+1}.
\end{align}
Adding Eq.~(\ref{discrete_entropy_13}) to Eq.~(\ref{discrete_entropy_12}), we obtain
\begin{align} \label{discrete_entropy_14}
\psi_{\bar{t}}\frac{\partial f}{\partial \psi}&=-\frac{\partial f}{\partial \psi}(\boldsymbol{v}^{n} \cdot \boldsymbol{\nabla} \psi^{n+1})-\frac{1}{\rm We}\lambda_{f}(T^{n})\delta_{\bar{t}} \notag\\
&+\frac{1}{\rm We}\lambda_{f}(T^{n})\epsilon\boldsymbol{\nabla}\cdot\left(\psi_{\bar{t}}\boldsymbol{\nabla}\psi^{n+1}+(\boldsymbol{\nabla} \psi^{n+1}\otimes \boldsymbol{\nabla} \psi^{n+1})\cdot \boldsymbol{v}^{n}\right) \notag\\
&-\frac{1}{\rm We}\lambda_{f}(T^{n})(\boldsymbol{v}^{n} \cdot \boldsymbol{\nabla} \delta^{n+1}) \notag\\
&-\frac{1}{\rm We}\lambda_{f}(T^{n})\epsilon\boldsymbol{\nabla} \boldsymbol{v}^{n}:(\boldsymbol{\nabla} \psi^{n+1}\otimes \boldsymbol{\nabla} \psi^{n+1}) \notag\\
&-\mu_{0}^{n+1}\psi^{n+1} (\boldsymbol{\nabla}\cdot \boldsymbol{v}^{n})-p^{n+1}(\boldsymbol{\nabla} \cdot \boldsymbol{v}^{n}) \notag\\
&+ \frac{1}{\mathrm{Pe}_{\psi}}\boldsymbol{\nabla} \cdot (m_{\psi}^{n}\mu_{c}^{n+1} \boldsymbol{\nabla}  \mu_{c}^{n+1} ) -\frac{m_{\psi}^{n}}{\mathrm{Pe}_{\psi}}|\boldsymbol{\nabla}\mu_{c}^{n+1}|^2 \notag\\
&-err_1 -err_2,
\end{align}
where the following identity is used
\begin{align}
&\frac{1}{\mathrm{Pe}_{\psi}}  \boldsymbol{\nabla} \cdot \left(m_{\psi}^{n} \boldsymbol{\nabla} {\mu}_{c}^{n+1}\right)\mu_{0}^{n+1}+\frac{ \alpha}{\mathrm{Pe}_{\psi}} \boldsymbol{\nabla} \cdot \left(m_{\psi}^{n} \boldsymbol{\nabla} {\mu}_{c}^{n+1}\right) p^{n+1} \notag\\
=&\frac{1}{\mathrm{Pe}_{\psi}}  \boldsymbol{\nabla} \cdot \left(m_{\psi}^{n} \boldsymbol{\nabla} {\mu}_{c}^{n+1}\right)\mu_{c}^{n+1}\notag\\
=&\frac{1}{\mathrm{Pe}_{\psi}}\boldsymbol{\nabla} \cdot (m_{\psi}^{n}\mu_{c}^{n+1} \boldsymbol{\nabla}  \mu_{c}^{n+1} ) -\frac{m_{\psi}^{n}}{\mathrm{Pe}_{\psi}}|\boldsymbol{\nabla}\mu_{c}^{n+1}|^2.\notag
\end{align}
Next, multiplying Eq.~(\ref{discrete_entropy_14}) by $-{\rm Ec}$, and adding to Eq.~(\ref{discrete_energy_conservation}), we obtain
%\begin{align}
%\frac{\partial u}{\partial T}T_{\bar{t}}+\frac{\partial u}{\partial \psi} \psi_{\bar{t}}=&-\frac{\partial u}{\partial T}(\boldsymbol{v}^{n}\cdot \boldsymbol{\nabla} T^{n})-\frac{\partial u}{\partial \psi}(\boldsymbol{v}^{n}\cdot \boldsymbol{\nabla} %\psi^{n+1})+ \Delta  T^{n+1}\\ \notag
%+&\lambda_{u}\epsilon \boldsymbol{\nabla} \cdot\left(\boldsymbol{\nabla}\psi^{n+1} \psi_{\bar{t}}+\boldsymbol{\nabla}\psi^{n+1} (\boldsymbol{v}^{n}\cdot \boldsymbol{\nabla} \psi^{n+1})\right )\\ \notag
%+&\boldsymbol{\nabla} \cdot (\mu_{c}^{n+1}\boldsymbol{\nabla} \mu_{c}^{n+1})-p^{n+1}(\boldsymbol{\nabla} \cdot \boldsymbol{v}^{n})\\ \notag
%-&T^{n}\hat{s}(\psi^{n+1},\boldsymbol{\nabla}\psi^{n+1})(\boldsymbol{\nabla} \cdot \boldsymbol{v}^{n})+\mathbf{\tau}^{n+1}:\boldsymbol{\nabla} \boldsymbol{v}^{n}\\ \notag
%-&\lambda_{u}\delta_{\bar{t}}-\lambda_{u}(\boldsymbol{v}^{n}\cdot \boldsymbol{\nabla} \delta^{n+1}) \\ \notag
%-&\lambda_{f}(T^{n}) \epsilon\boldsymbol{\nabla} \boldsymbol{v}^{n}:(\boldsymbol{\nabla} \psi^{n+1}\otimes \boldsymbol{\nabla} \psi^{n+1})\\ \notag
%-&\boldsymbol{\nabla} \boldsymbol{v}^{n}:\mu_{0}^{n+1}\psi^{n+1}\mathbf{I}-error^{n},
%\end{align}
%
    \begin{align}
    \label{discrete_entropy_15}
(\rho C_h)^{n+1} T_{\bar{t}}+{\rm Ec}\frac{\partial u}{\partial \psi} \psi_{\bar{t}}-{\rm Ec}\frac{\partial f}{\partial \psi}\psi_{\bar{t}}&=-(\rho C_h)^{n+1}(\boldsymbol{v}^{n}\cdot \boldsymbol{\nabla} T^{n})
    \notag
    \\
&-{\rm Ec}\frac{\partial \tilde{u}}{\partial \psi}(\boldsymbol{v}^{n}\cdot \boldsymbol{\nabla} \psi^{n+1})+{\rm Ec}\frac{\partial f}{\partial \psi}(\boldsymbol{v}^{n} \cdot \boldsymbol{\nabla} \psi^{n+1})
    \notag
    \\
&-\frac{{\rm Ec}}{{\rm We}}T^{n}\lambda_{s}\delta_{\bar{t}}+ \frac{1}{\mathrm{Pe}_{T}} \boldsymbol{\nabla}\cdot\left(k \frac{\boldsymbol{\nabla} T^{n+1}+\boldsymbol{\nabla} T^{n}}{2}\right)
    \notag
    \\
&+\frac{{\rm Ec}}{{\rm We}}T^{n}\lambda_{s}\epsilon \boldsymbol{\nabla} \cdot \left(\boldsymbol{\nabla}\psi^{n+1} \psi_{\bar{t}}+\boldsymbol{\nabla}\psi^{n+1} (\boldsymbol{v}^{n}\cdot \boldsymbol{\nabla} \psi^{n+1}) \right)
    \notag
    \\
&-\frac{{\rm Ec}}{{\rm We}}T^{n}\lambda_{s}(\boldsymbol{v}^{n}\cdot \boldsymbol{\nabla} \delta^{n+1})
    \notag
    \\
&-{\rm Ec}~T^{n}\tilde{s}^{n+1}(\boldsymbol{\nabla} \cdot \boldsymbol{v}^{n})+\frac{{\rm Ec}}{\mathrm{Re}}\mathbf{\tau}^{n+1}:\boldsymbol{\nabla} \boldsymbol{v}^{n+1}
    \notag
    \\
& +\frac{{\rm Ec}}{\mathrm{Pe}_{\psi}}m_{\psi}^{n}|\boldsymbol{\nabla}\mu_{c}^{n+1}|^2+corr^{n}
    \notag
    \\
&+{\rm Ec}~err_1 +{\rm Ec}~err_2,
    \end{align}
where the following identities are used
\begin{align}
T^{n}\lambda_{s}&=\lambda_{u}-\lambda_{f}(T^{n}),\\
(\boldsymbol{\nabla} \psi^{n+1}\otimes \boldsymbol{\nabla} \psi^{n+1})\cdot \boldsymbol{v}^{n}&=\boldsymbol{\nabla}\psi^{n+1} (\boldsymbol{v}^{n}\cdot \boldsymbol{\nabla} \psi^{n+1}),\\
\boldsymbol{\nabla} \boldsymbol{v}^{n}:\mu_{0}^{n+1}\psi^{n+1}\mathbf{I}&=\mu_{0}^{n+1}\psi^{n+1} (\boldsymbol{\nabla}\cdot \boldsymbol{v}^{n}).
\end{align}
Furthermore, using Eqs.~(\ref{discrete_f_psi}) and (\ref{discrete_u_psi}), we obtain
\begin{align} \label{discrete_entropy_16}
{\rm Ec}\frac{\partial f}{\partial \psi} \psi_{\bar{t}}&=\rho_{\bar{t}} C_h^{n+1}T^{n}-\rho_{\bar{t}} C_h^{n+1}T^{n}\ln \left(\frac{T^{n}}{T_0}\right) \notag\\
& + \rho^{n} (C_h)_{\bar{t}}T^{n}-\rho^{n} (C_h)_{\bar{t}}T^{n}\ln \left(\frac{T^{n}}{T_0}\right)\notag \\
&= (\rho C_h)_{\bar{t}}T^{n}-(\rho C_h)_{\bar{t}}T^{n}\ln \left(\frac{T^{n}}{T_0}\right),
\end{align}
and
\begin{align} \label{discrete_entropy_17}
{\rm Ec}\frac{\partial u}{\partial \psi}\psi_{\bar{t}}+(\rho C_h)^{n+1} T_{\bar{t}}  &= (\rho C_h)_{\bar{t}}T^{n} \notag\\
&+ \frac{(\rho C_h)^{n+1}T^n}{\delta t}\left(\ln(\frac{T^{n+1}}{T_0}) - \ln(\frac{T^{n}}{T_0})\right) \notag\\
&+err_{3},\\
err_{3}&=\frac{(\rho C_h)^{n+1}T^{n} (T^{n+1}-T^{n})^{2}}{2(\tilde{T})^{2} \delta t}, \label{err_3}
\end{align}
where $\tilde{T}$ is between $T^{n+1}$ and $T^{n}$, and we have used the following identity
\begin{align}
T^{n+1}-T^{n}=&\left(\ln(\frac{T^{n+1}}{T_0})-\ln(\frac{T^{n}}{T_0})\right)T^{n}+\frac{T^{n}}{2(\tilde{T})^{2}}(T^{n+1}-T^{n})^{2}.
\end{align}
Subtracting Eq.~(\ref{discrete_entropy_16}) from Eq.~(\ref{discrete_entropy_17}), we obtain
\begin{align} \label{discrete_entropy_18}
(\rho C_h)^{n+1} T_{\bar{t}}+{\rm Ec}\frac{\partial u}{\partial \psi} \psi_{\bar{t}}-{\rm Ec}\frac{\partial f}{\partial \psi}\psi_{\bar{t}}&=
 \frac{(\rho C_h)^{n+1}T^n}{\delta t}\ln(\frac{T^{n+1}}{T_0}) \notag\\
 &-\frac{(\rho C_h)^{n}T^{n}}{\delta t}\ln (\frac{T^{n}}{T_0})+err_{3} \notag\\
&={\rm Ec}~T^{n}s_{\bar{t}}+err_{3}.
\end{align}
In addition, we have
\begin{align}\label{discrete_entropy_19}
&-(\rho C_h)^{n+1}(\boldsymbol{v}^{n}\cdot \boldsymbol{\nabla} T^{n})-{\rm Ec}\frac{\partial \tilde{u}}{\partial \psi}(\boldsymbol{v}^{n}\cdot \boldsymbol{\nabla} \psi^{n+1})+{\rm Ec}\frac{\partial f}{\partial \psi}(\boldsymbol{v}^{n} \cdot \boldsymbol{\nabla} \psi^{n+1}) \notag\\
=&-{\rm Ec}~T^{n}(\boldsymbol{v}^{n}\cdot \boldsymbol{\nabla} \tilde{S}^{n+1}),
\end{align}
where the following identities are used
\begin{align}
-(\rho C_h)^{n+1}(\boldsymbol{v}^{n}\cdot \boldsymbol{\nabla} T^{n})&=-{\rm Ec}~T^{n}\frac{\partial \tilde{S}^{n+1}}{\partial T^{n}}(\boldsymbol{v}^{n}\cdot \boldsymbol{\nabla} T^{n}),\\
-\frac{\partial \tilde{u}}{\partial \psi}(\boldsymbol{v}^{n}\cdot \boldsymbol{\nabla} \psi^{n+1})+\frac{\partial f}{\partial \psi}(\boldsymbol{v}^{n} \cdot \boldsymbol{\nabla} \psi^{n+1})&=-T^{n}\frac{\partial \tilde{S}^{n+1}}{\partial \psi^{n+1}}(\boldsymbol{v}^{n} \cdot \boldsymbol{\nabla} \psi^{n+1}).
\end{align}
Substituting Eqs.~(\ref{discrete_entropy_18}) and (\ref{discrete_entropy_19}) into Eq.~(\ref{discrete_entropy_15}), and recalling the definitions of $\hat{s} $ (Eq.~(\ref{non-dimensional-entropy})) and $\tilde{s}^{n+1}$ (Eq.~(\ref{tilde_s})), we obtain
    \begin{align}
    \label{discrete_entropy_20_2}
{\rm Ec}~T^{n}\hat{s}_{\bar{t}}&=\frac{1}{\mathrm{Pe}_{T}} \boldsymbol{\nabla}\cdot\left(k \frac{\boldsymbol{\nabla} T^{n+1}+\boldsymbol{\nabla} T^{n}}{2}\right)
    \notag
    \\
&+\frac{{\rm Ec}}{{\rm We}}T^{n}\lambda_{s}\epsilon \boldsymbol{\nabla} \cdot \left(\boldsymbol{\nabla}\psi^{n+1} \psi_{\bar{t}}+\boldsymbol{\nabla}\psi^{n+1} (\boldsymbol{v}^{n}\cdot \boldsymbol{\nabla} \psi^{n+1}) \right)
    \notag
    \\
&-{\rm Ec}~T^{n}\boldsymbol{\nabla}\cdot (\tilde{s}^{n+1}{v}^{n})+\frac{{\rm Ec}}{\mathrm{Re}}\mathbf{\tau}^{n+1}:\boldsymbol{\nabla} \boldsymbol{v}^{n+1}
    \notag
    \\
& +\frac{{\rm Ec}}{\mathrm{Pe}_{\psi}}m_{\psi}^{n}|\boldsymbol{\nabla}\mu_{c}^{n+1}|^2+corr^{n}+{\rm Ec}~err_1 +{\rm Ec}~err_2-err_3.
    \end{align}
Multiplying Eq. (\ref{discrete_entropy_20_2}) by $1/({\rm Ec}~T^{n})$, and using the following identities
\begin{align}
\frac{1}{T^{n}} \boldsymbol{\nabla}\cdot\left(k \frac{\boldsymbol{\nabla} T^{n+1}+\boldsymbol{\nabla} T^{n}}{2}\right)
&=\boldsymbol{\nabla} \cdot(\frac{k}{T^{n}} \frac{\boldsymbol{\nabla}  T^{n+1}+\boldsymbol{\nabla}T^{n}}{2})+\frac{k|\boldsymbol{\nabla}  T^{n+1}+\boldsymbol{\nabla}T^{n}|^2}{4(T^{n})^2} \notag\\
&+\frac{err_{4}}{T^{n}}, \\
err_{4}&=\frac{k|\boldsymbol{\nabla}T^{n}|^2-k|\boldsymbol{\nabla} T^{n+1}|^{2}}{4T^{n}}, \label{err_4}
\end{align}
we obtain
\begin{align}  \label{discrete_entropy_21}
\hat{s}_{\bar{t}}&=\frac{k}{{\rm Ec}~\mathrm{Pe}_{T}}\frac{|\boldsymbol{\nabla}  T^{n+1}+\boldsymbol{\nabla}T^{n}|^2}{4(T^{n})^2} \notag\\
&+\frac{1}{{\rm Ec}~\mathrm{Pe}_{T}}\boldsymbol{\nabla} \cdot(\frac{k}{T^{n}} \frac{\boldsymbol{\nabla}  T^{n+1}+\boldsymbol{\nabla}T^{n}}{2}) \notag\\
&+\frac{1}{\rm We}\lambda_{s}\epsilon \boldsymbol{\nabla} \cdot \left(\boldsymbol{\nabla}\psi^{n+1} \psi_{\bar{t}}+\boldsymbol{\nabla}\psi^{n+1} (\boldsymbol{v}^{n}\cdot \boldsymbol{\nabla} \psi^{n+1}) \right) \notag\\
&-\boldsymbol{\nabla}\cdot (\tilde{s}^{n+1}{v}^{n})+\frac{1}{\mathrm{Re}}\frac{\mathbf{\tau}^{n+1}:\boldsymbol{\nabla} \boldsymbol{v}^{n+1}}{T^{n}}+\frac{m_{\psi}^{n}}{\mathrm{Pe}_{\psi}}\frac{|\boldsymbol{\nabla}\mu_{c}^{n+1}|^2}{T^{n}}+res,
\end{align}
where
\begin{align}
res=\frac{corr^{n}+{\rm Ec}~err_1 +{\rm Ec}~err_2-err_3+\frac{1}{\mathrm{Pe}_{T}}err_4}{{\rm Ec}~T^{n}}.
\end{align}
Recalling the definitions of $corr^{n}$ (Eq.~({\ref{corr}})), $err_{1}$ (Eq.~({\ref{err_1}})), $err_{2}$ (Eq.~({\ref{err_2}})), $err_{3}$ (Eq.~({\ref{err_3}})) and $err_{4}$ (Eq.~({\ref{err_4}})), we obtain
\begin{align} \label{res}
res&=\frac{{\rm Ec}}{\rm We}\lambda_{f}(T^{n})\left(\frac{1}{4}+W''(\xi)\right)\frac{
({\psi}^{n}-\psi^{n+1})^2}{2 T^{n} \delta t \epsilon}+\frac{(\rho C_h)^{n+1} (T^{n+1}-T^{n})^{2}}{2(T_{min})^{2} \delta t} \notag\\
&-\frac{(\rho C_h)^{n+1} (T^{n+1}-T^{n})^{2}}{2 (\tilde{T})^{2} \delta t}\geq 0,
\end{align}
where we have used the relations $W''(\xi)\geq -1/{4}$, $\lambda_{f}(T^{n})\geq 0$ and $T_{min}\leq \tilde{T}$.\\
Taking the integral over $\Omega$  and applying the divergence theorem with the boundary conditions in Eq.~(\ref{proof_bc_discrete}), we finally deduce the discrete entropy increase (\ref{dis-mass-law}) of our numerical method
\begin{align}  \label{discrete_entropy_22}
S_{\bar{t}} =\int_{\Omega}\hat{s}_{\bar{t}}~d{x}=& \int_{\Omega}\bigg\{\frac{k}{{\rm Ec}~\mathrm{Pe}_{T}}\frac{|\boldsymbol{\nabla}  T^{n+1}+\boldsymbol{\nabla}T^{n}|^2}{4(T^{n})^2}+\frac{1}{\mathrm{Re}}\frac{\mathbf{\tau}^{n+1}:\boldsymbol{\nabla} \boldsymbol{v}^{n+1}}{T^{n}} \notag\\
&+\frac{m_{\psi}^{n}}{\mathrm{Pe}_{\psi}}\frac{|\boldsymbol{\nabla}\mu_{c}^{n+1}|^2}{T^{n}}+res\bigg\} d {x} \geq0.
\end{align}

\end{proof}

\section{Numerical results} \label{numerical result}

\subsection{Time accuracy test}\label{convergence_test}
We first conduct a time accuracy test for our numerical method (\ref{discrete_momentum-conservation}) -- (\ref{discrete_energy_conservation}) in a 2D domain $[0,2]\times[0,2]$, while also demonstrating that the entropy remains increasing. Here, we fix the grid size as $1024\times1024$, ensuring that errors from spatial discretization are negligible compared to time discretization errors. Without considering gravity, we set the non-dimensional parameters as follows
\begin{align}
&\mathrm{Pe}_{\psi}=100, \mathrm{Re}=1, {\rm We}=10, {\rm Ca}=10, {\rm Ma}=0.01, \mathrm{Pe}_{T}=1,{\rm Ec}=1,\notag\\
&\zeta_{\rho}=1, \zeta_{\mu}=0.1, \zeta_{C_h}=0.1, \zeta_{k}=0.1,\epsilon=0.05,\eta=6 \sqrt{2}.
\end{align}
The initial conditions are given as
\begin{align}
\psi(x,y,0)&=0.5\left(\sin (\pi x) \cos (\pi y)+1\right), \\
\boldsymbol{v}(x,y,0)&=\left(u_0, v_0\right)=\left(\sin (\pi x) \sin (\pi y),\sin (\pi x) \sin (\pi y)\right),\\
T(x,y,0)&=0.5\left(\cos (\pi x) \cos (\pi y)+1\right),\\
 p(x,y,0)&=\sin (\pi x) \sin (\pi y).
\end{align}
On the left and right boundaries of the domain, we apply periodic boundary conditions. On the top and bottom walls, we apply no-slip boundary condition for $\boldsymbol{v}$ and no-flux boundary conditions for $\psi$, $\mu_{c}$ and $T$.
As exact solutions are not available, errors in $L^2$ norms are calculated as the difference between the solution of the coarse time step and that of the adjacent finer time step. The errors of the phase variable $\psi$ and velocity components $u$ and $v$ at $t=0.25$ with various time step sizes are presented in Tab.~\ref{tab:1}. We observe that our numerical method (\ref{discrete_momentum-conservation}) -- (\ref{discrete_energy_conservation}) almost perfectly matches the first-order accuracy in time. Additionally, the numerical results of the entropy and volume of the two-phase fluid system with various time steps are shown in Figs.~\ref{entropy} and \ref{volume}, where we observe that the entropy remains increasing, and the volume is preserved up to $10^{-12}$ (all the lines in Fig.~\ref{volume} are coincident).

\begin{table}[ht]
\centering
\caption{Comparison of $L^2$ errors for the phase variable $\psi$, and velocity components $u$ and $v$ obtained at $t = 0.25$ with various $\delta t$ on a fixed grid size of $1024\times1024$. Refer to \S\ref{convergence_test} for details. \label{tab:1}}
\renewcommand\arraystretch{1.25}
\begin{tabular}{ccccccc}
\hline
        $\delta t$ & $L^2$ error of $\psi$ & Order & $L^2$ error of $u$ & Order & $L^2$ error of $v$ & Order\\
        \hline
       % $1/32$ &   &  &  &  & &  \\
        $1/64$ & 4.2312$\rm{e}$-2  &  & 9.7261$\rm{e}$-3 &  &1.0042$\rm{e}$-2 & \\
        $1/128$ & 2.1379$\rm{e}$-2 & 0.9849  & 4.7447$\rm{e}$-3  &1.0355 &4.9610$\rm{e}$-3 & 1.0174\\
        $1/256$ & 1.0806$\rm{e}$-2 & 0.9844 & 2.3420$\rm{e}$-3  &1.0186 &2.4671$\rm{e}$-3 & 1.0078\\
        $1/512$ & 5.4513$\rm{e}$-3 & 0.9871 & 1.1675$\rm{e}$-3 &1.0043  &1.2358$\rm{e}$-3 &0.9974 \\
        $1/1024$ & 2.7434$\rm{e}$-3 & 0.9906 & 5.8559$\rm{e}$-4 &0.9955  &6.2164$\rm{e}$-4 &0.9913 \\
        \hline
\end{tabular}
\end{table}

\begin{figure}[ht]
\centering
\includegraphics[scale=1]{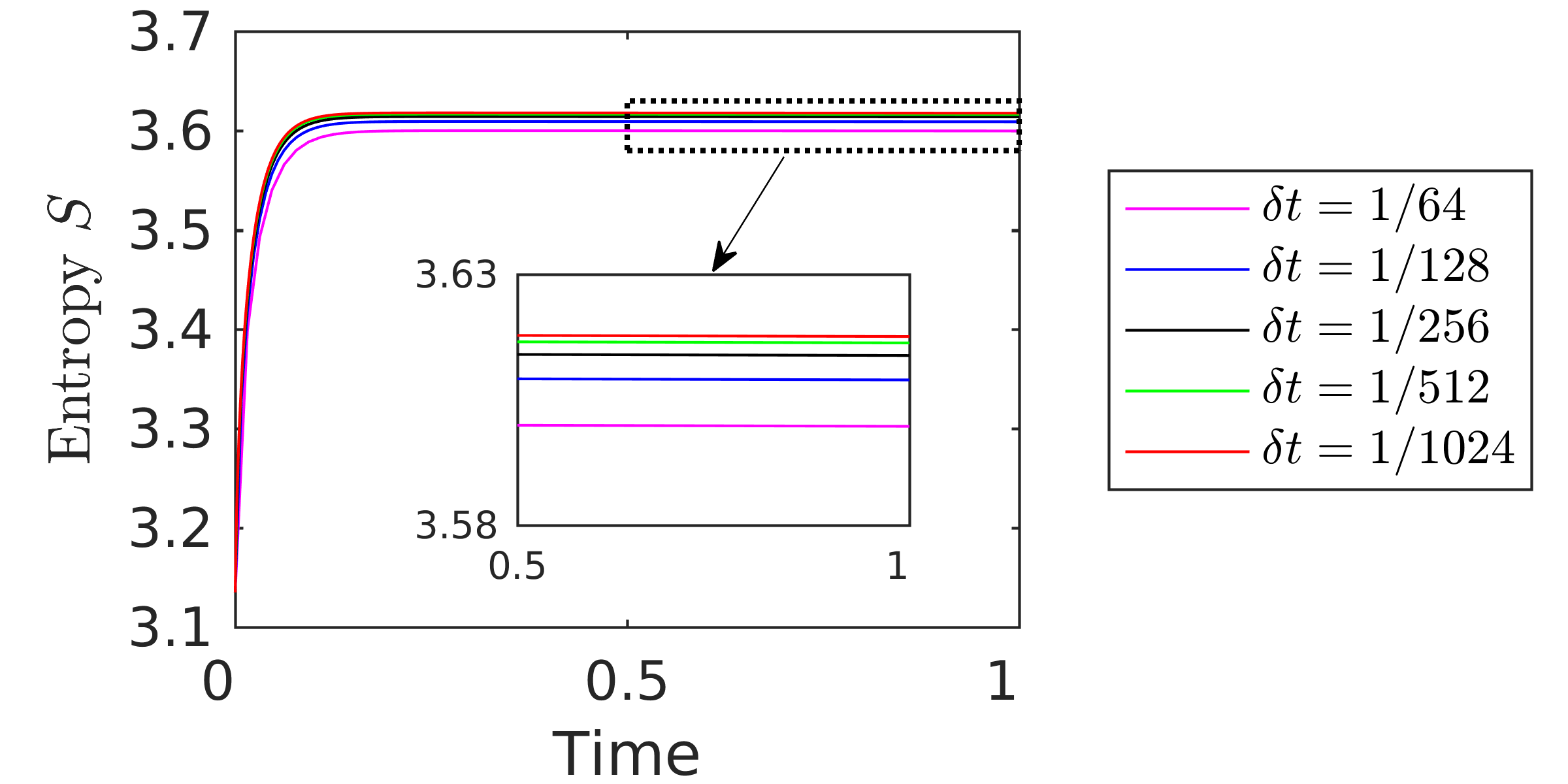}%%\label{fig:25:a}//
\caption{Entropy ($S$) computed using Eq.~(\ref{non-dimension-total-entropy}) with different time steps on a fixed grid size of $1024\times1024$. Refer to \S\ref{convergence_test} for details.}       \label{entropy}
\end{figure}

\begin{figure}[ht]
\centering
\includegraphics[scale=1]{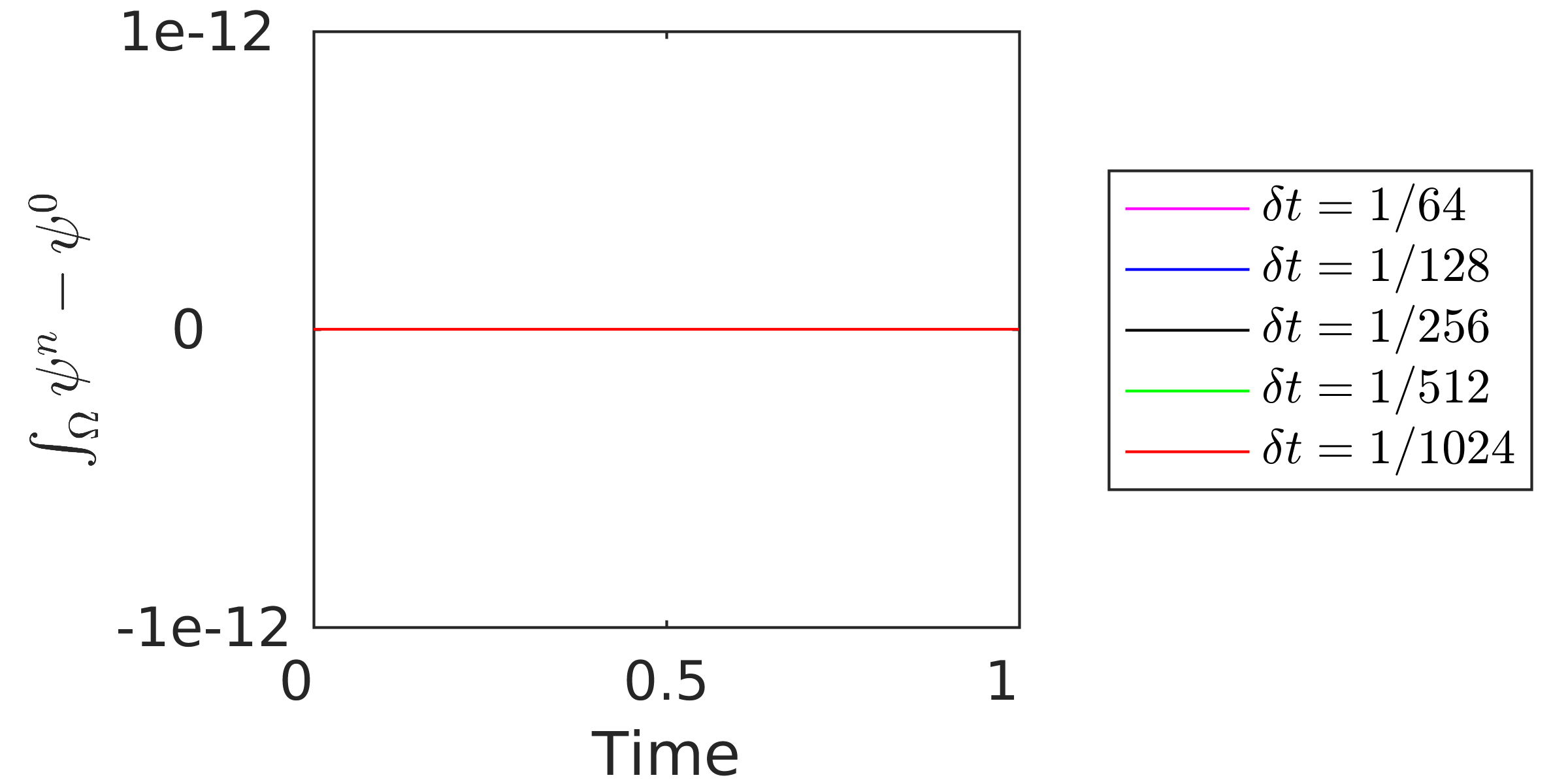}%%\label{fig:25:a}//
\caption{Discrete volume conservation $\int_{\Omega}\psi^{n}-\psi^{0}$ with different time steps on a fixed grid size of $1024\times1024$. Refer to \S\ref{convergence_test} for details.}       \label{volume}
\end{figure}

\subsection{Thermocapillary migration of a droplet with infinitely small $\mathrm{Pe}_{T}$}\label{Thermocapillary-migration-of-a-drop}
In this subsection, we explore the thermocapillary migration of a droplet in a square microchannel with a fixed linearly increasing temperature field imposed along the channel. The temperature gradient is ${\partial T}/{\partial z}=\nabla T_{\infty}>0$, and the effect of gravity is assumed to be negligible.

The droplet (fluid 2) of radius $R$ is surrounded by another immiscible fluid (fluid 1) in the microchannel. The thermocapillary migration of a droplet was first investigated analytically by Young \emph{et al.} \cite{young1959motion}, where both the heat Peclect and Reynolds numbers are assumed to be infinitely small, and the convective transport of momentum and energy is neglected. The terminal velocity (also known as YGB velocity) of the droplet under constant temperature gradient $\nabla T_{\infty}$ is given as
\begin{align}
V_{YGB} = \frac{2\sigma_T \nabla T_{\infty} R}{(2+k_1/k_2)(2\mu_1+3\mu_2)}.
\end{align}

To validate our phase-field model, we compare the numerical results to the analytical solution $V_{YGB}$ obtained from the sharp-interface model. The characteristic length, velocity and temperature for this test are $R^*=10R$, $V^* = V_{YGB}$ and $T^{*}=\nabla T_{\infty} R$, respectively. Numerical simulations are carried out in a 3D dimensionless domain $[0,0.75]\times[0,0.75]\times[0,1.5]$. We impose periodic boundary conditions on all the side boundaries ($x=0,0.75$,~$y=0,0.75$), and apply no-slip boundary condition for $\boldsymbol{v}$ and no-flux boundary conditions for $\psi$ and $\mu_{c}$ on the top and bottom boundaries ($z=0,1.5$).

The droplet of dimensionless radius $0.1$ is initially  stationary and centered at (0.375,0.375,0.75). Therefore, the initial conditions for velocity and phase variable are given as
\begin{align}
\boldsymbol{v}(x,y,z,0) &= 0,\\
\psi(x,y,z,0) &= 0.5 + 0.5\tanh{\left({(0.1-r)}/{2\sqrt{2}\epsilon}\right)},
\end{align}
where $r = \sqrt{(x-0.375)^2+(y-0.375)^2+(z-0.75)^2}$.
In addition, the dimensionless temperature field is fixed as
\begin{align}
T(z) = 10z+1,
\end{align}
and thus, the heat Eq. (\ref{sys-energy-nodim}) is not considered. The model parameters and ratios of physical properties are set as follows %the numerical simulation is carried out with
\begin{align}
&\mathrm{Pe}_{\psi}=1.5\times10^{4}/\epsilon, \mathrm{Re}=1.\dot{3}\times 10^{-3}, {\rm We}=8.\dot{8}\times 10^{-6}, {\rm Ca}=6.\dot{6}\times 10^{-3}, \notag\\ &{\rm Ma}=7.5, {\rm Ec}=1.\dot{7}\times 10^{-6},  \zeta_{\rho}=1, \zeta_{\mu}=1, \zeta_{C_h}=1, \eta=6 \sqrt{2}.
\end{align}
In the simulations, the droplet migration velocity, $v_{d}$, is calculated numerically by % $\sigma_0= 0.2$ (at $T_0=0$),  the fluid properties of
\begin{align}
v_{d}=\frac{\int_{\Omega} \psi \boldsymbol{v}\cdot \boldsymbol{k} {\mathrm d}\bs{x}}{\int_{\Omega} \psi {\rm d}\bs{x}},
\end{align}
where $\boldsymbol{k}$ is the unit vector in $z$ direction. We first demonstrate the convergence of the model by testing various values of $\epsilon(=0.0025, 0.005, 0.01, 0.02)$ with a fixed grid size $128\times128\times256$ and a time step of $10^{-5}$. As shown in Fig.~\ref{fix_grid}(a), the terminal velocity converges to $V_{YGB}$ asymptotically as the value of $\epsilon$ decreases. Next, we show the convergence of the results by refining the grid with a fixed $\epsilon=0.0025$ and a time step of $10^{-5}$. Specially, we use four grid sizes ($32\times32\times64, 64\times64\times128, 128\times128\times256, 256\times256\times512$) for the computations. As shown in Fig.~\ref{fix_grid}(b), the migration velocity converges as the grid size increases, and the results for the grid sizes $128\times128\times256$ and $256\times256\times512$ are very close. Therefore, we set the grid size $128\times128\times256$ and a time step of $10^{-5}$ for the 3D computations in the subsequent section.\\ %\textcolor{red}{ $\epsilon=0.0025$}

%\begin{figure}[!h]
%\centering
%\includegraphics[scale=1]{differenteps_velocity.png}
%\caption{Time evolution of the thermocapillary migration velocity of a droplet with various $\epsilon$ and a fixed grid size $128\times128\times256$. } \label{different_eps}
%\end{figure}

\begin{figure}[ht]
\centering
\subfigure[]{\includegraphics[scale=0.9]{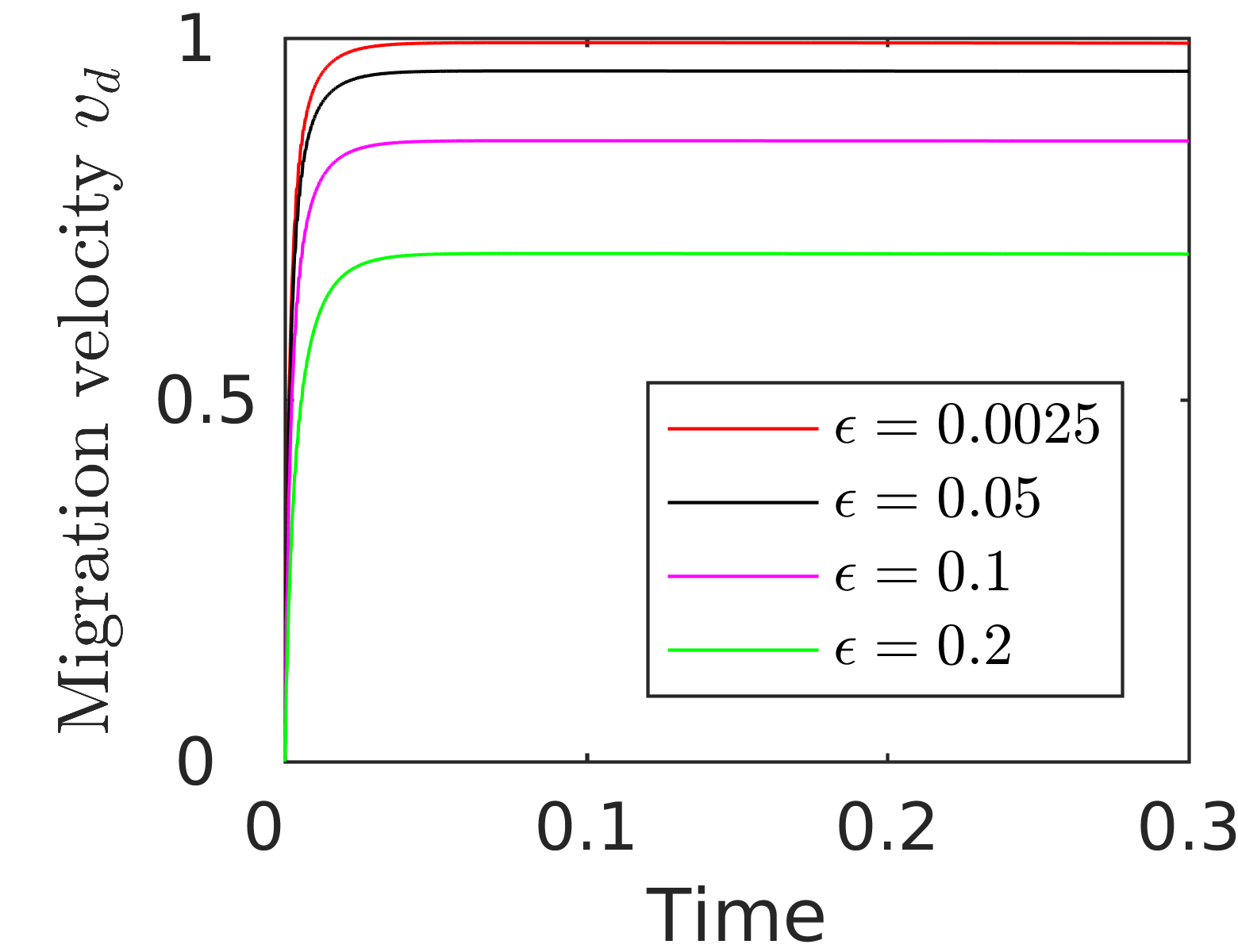}}%%\label{fig:25:a}//
\subfigure[]{\includegraphics[scale=0.9]{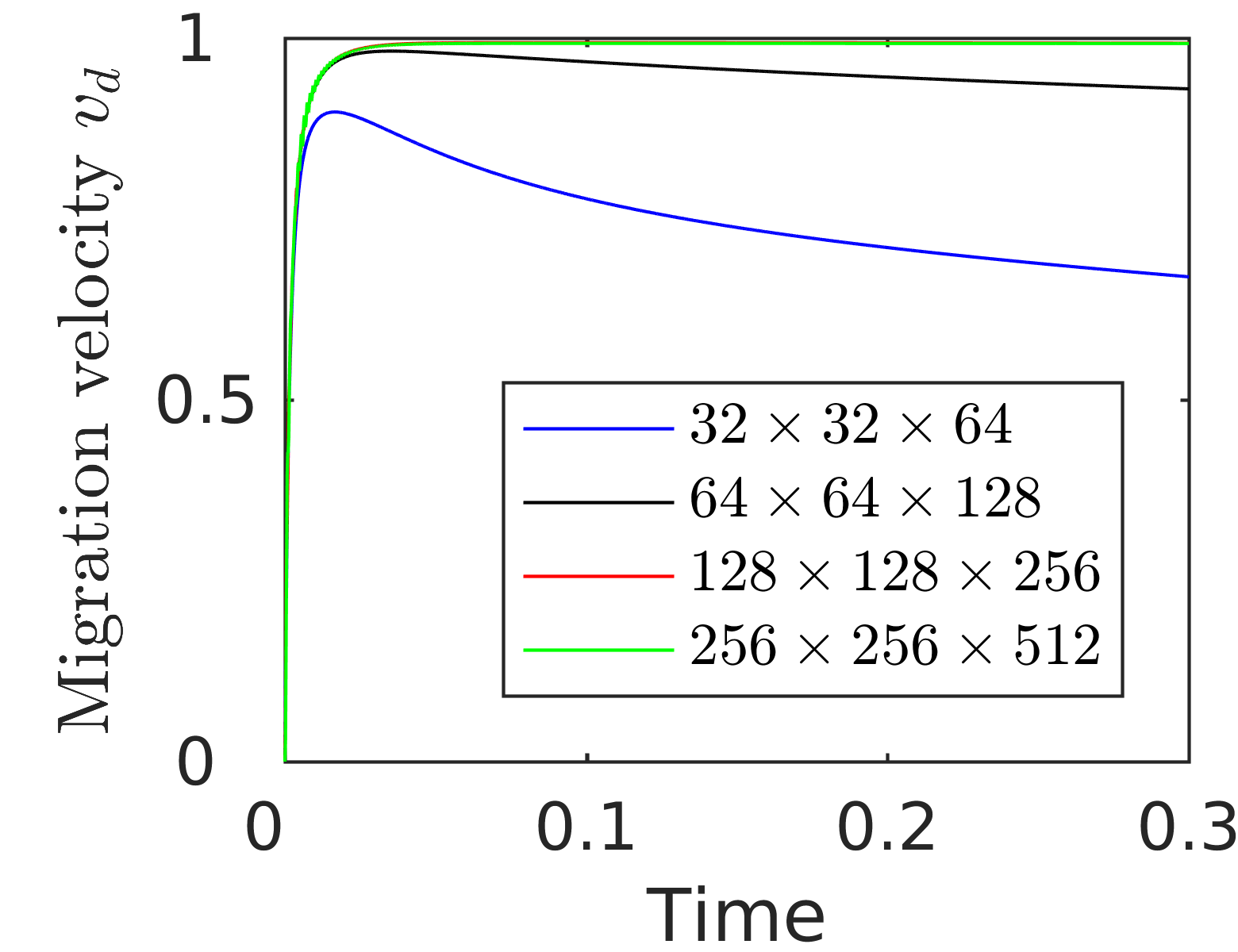}}
\caption{Time evolution of the thermocapillary migration velocity of a droplet with (a) various $\epsilon$ and a fixed grid size $128\times128\times256$, and (b) various grid sizes and a fixed $\epsilon=0.0025$. See \S\ref{Thermocapillary-migration-of-a-drop} for details.}        \label{fix_grid}
\end{figure}

\subsection{Thermocapillary-driven convection }\label{Thermocapillary-convection-of-two-planar-fluid}
%\textcolor{blue}{(Here the  dimensional equations are used.)}

We now investigate the thermocapillary-driven convection in a heated microchannel with two superimposed planar fluids \cite{pendse2010analytical}. The setup of the problem is illustrated in Fig.~\ref{diagram_planar_fluids}. The heights of the fluid 1 (upper) and fluid 2 (lower) are $a$ and $b$, respectively, and the fluids extend infinitely in the horizontal direction. The temperatures of the upper and lower walls are imposed as follows:
\begin{align}
T(x,a) = T_{c}, \label{example_2_topT}
\end{align}
and
\begin{align}
T (x,-b) = T_{h} + T_{0} cos(\omega x),\label{example_2_lowerT}
\end{align}
respectively, where ${T_{h}>T_{c}>T_{0}>0}$, and ${\omega=2 \pi / l}$ is a wave number with ${l}$ being the channel length. The given thermal boundary conditions establish a temperature field that is periodic in the horizontal direction with a period length of ${l}$. Therefore, it is sufficient to only concentrate on the solution within one period domain ${-l / 2<x<l / 2}$.

When the heat Peclect number and Reynolds number are negligibly small, it is possible to ignore the convective transport of momentum and energy. In addition, it is assumed that the interface remains flat, and the effect of gravity is assumed to be negligible. By solving the simplified linear governing equations, Pendse and Esmaeeli \cite{pendse2010analytical} obtained the analytical solutions for the temperature field ${\bar{T}(x, y)}$ and stream-function ${\bar{\Phi}(x, y)}$. Specifically, for the upper fluid, the solutions are obtained as follows:
\begin{align}
\bar{T}(x, y)=&\frac{\left(T_{c}-T_{h}\right) y+\tilde{k} T_{c} b+T_{h} a}{a+\tilde{k} b}+T_{0} f(\alpha, \beta, \tilde{k}) \sinh (\alpha-\omega y) \cos (\omega x),\\
\bar{\Phi}(x, y)=&\frac{U_{\max }}{\omega} \frac{1}{\sinh ^{2}(\alpha)-\alpha^{2} }\{\omega y \sinh ^{2}(\alpha) \cosh (\omega y) \notag\\
& -\frac{1}{2}\left[2 \alpha^{2}+\omega y(\sinh (2 \alpha)-2 \alpha)\right] \sinh (\omega y)\} \sin (\omega x),
\end{align}
and for the lower fluid
\begin{align}
\bar{T}(x, y)=&\frac{\tilde{k}\left(T_{c}-T_{h}\right) y+\tilde{k} T_{c} b+T_{h} a}{a+\tilde{k} b}+T_{0} f(\alpha, \beta, \tilde{k})[\sinh (\alpha) \cosh (\omega y)\notag\\
&-\tilde{k} \sinh (\omega y) \cosh (\alpha)] \cos (\omega x), \\
\bar{\Phi}(x, y)=&\frac{U_{\max }}{\omega} \frac{1}{\sinh ^{2}(\beta)-\beta^{2}} \{\omega y \sinh ^{2}(\beta) \cosh (\omega y)\notag\\
&-\frac{1}{2}\left[2 \beta^{2}-\omega y(\sinh (2 \beta)-2 \beta)\right] \sinh (\omega y)\} \sin (\omega x).
\end{align}
The unknowns in the above equations are defined by
\begin{align}
\tilde{k}=\frac{k_{1}}{k_{2}}, ~\alpha=a \omega, ~\beta=b \omega,~f(\alpha, \beta, \tilde{k})=\frac{1}{\tilde{k} \sinh (\beta) \cosh (\alpha)+\sinh (\alpha) \cosh (\beta)},
\end{align}
and
\begin{align}
U_{\max }=\left(\frac{T_{0} \sigma_{T}}{\mu_{2}}\right) g(\alpha, \beta, \tilde{k}) h(\alpha, \beta, \tilde{\mu}),
\end{align}
where the subscripts 1 and 2 stand for the fluid 1 and 2, respectively,
\begin{align}
g(\alpha, \beta, \tilde{k})=&\sinh (\alpha) f(\alpha, \beta, \tilde{k}),\\
h(\alpha, \beta, \tilde{\mu})=&\frac{\left(\sinh ^{2}(\alpha)-\alpha^{2}\right)\left(\sinh ^{2}(\beta)-\beta^{2}\right)}{\tilde{\mu}\left(\sinh ^{2}(\beta)-\beta^{2}\right)(\sinh (2 \alpha)-2 \alpha)+\left(\sinh ^{2}(\alpha)-\alpha^{2}\right)(\sinh (2 \beta)-2 \beta)},
\end{align}
$ \tilde{k}=k_{1}/k_2$ and $\tilde{\mu}=\mu_{1}/\mu_2$ are the thermal conductivity ratio and viscosity ratio between the two fluids, respectively.

Here, we use $a$, $T_{0} \sigma_{T}/\mu_{2}$, $T_{c}$ and $T_{0} \sigma_{T} a/\mu_{2}$ as the characteristic length, velocity, temperature and stream-function for the present test, respectively. The numerical simulations are carried out in a 2D domain $[-l/2, l/2] \times [-b, a]$, where
\begin{align}
l=2, \quad a=b=1,
\end{align}
resulting in a dimensionless domain $[-1, 1] \times [-1, 1]$.
To ensure a flat and rigid interface between the two fluids, we specify the phase variable as
\begin{align}
\psi(y)=\frac{1}{2}+\frac{1}{2} \tanh \left(\frac{y}{2 \sqrt{2} \epsilon}\right), ~~~{\rm for}~~~  y \in(-1, 1),
\end{align}
where $\epsilon$ represents the thickness of the diffuse interface. Periodic boundary conditions are enforced on the left and right boundaries of the domain. On the upper and lower walls, no-slip boundary conditions are imposed. The wall temperatures are prescribed through Eqs.~(\ref{example_2_topT}) and (\ref{example_2_lowerT}), where dimensional temperatures are $T_{h}=20$, $T_{c}=10$, and $T_{0}=4$, yielding dimensionless temperatures $T_{h}=2$, $T_{c}=1$, and $T_{0}=0.4$. The ratios of fluid properties and model parameters are specified as:
\begin{align}
&\mathrm{Re}=0.05, {\rm We}=0.004, {\rm Ca}=0.08, {\rm Ma}=2.5, \mathrm{Pe}_{T}=0.01,\notag\\
& \zeta_{\rho}=1, \zeta_{\mu}=1, \zeta_{C_h}=1, \zeta_{k}=1/\tilde{k},\eta=6 \sqrt{2}.
\end{align}
 
To show the effect of the thermal conductivity ratio on the stream-function and the temperature field, we examine two cases with different values of $\tilde{k}$, i.e., $\tilde{k}=1$ for case 1, and $\tilde{k}=0.2$ for case 2. According to the definition of the variable thermal conductivity $k(\psi)$, we have a constant $k$ for case 1,
$$
k(\psi)={\psi+\frac{k_2}{k_1}(1-\psi)}=1,$$
and a variable $k(\psi)$ for case 2,
$$
k(\psi)={\psi+\frac{k_2}{k_1}(1-\psi)}=5-4\psi.
$$
Figs.~\ref{T_countor} and \ref{stramline} depict the temperature field and stream-functions contours for the two cases, respectively. The numerical simulation is conducted on a grid size $1024\times1024$ with a time step of $10^{-4}$ and $\epsilon=0.0025$. It's apparent that the numerical results agree well with the analytical solutions. Moreover, as the thermal conductivity ratio $\tilde{k}$ decreases, the temperature distribution at the interface becomes more nonuniform (See Fig.~\ref{T_countor}). This results in an enhanced shear force along the interface, and thus strengthens the thermocapillary-driven convection, as reflected in Fig.~\ref{stramline}, where the gradient of the stream-function increases as $\tilde{k}$ decreases. Additionally, to demonstrate the convergence of our phase-field model towards the sharp-interface model as the diffuse interface thickness tends to zero, we compute the $L^2$ norms of the relative differences between the numerical results with five different $\epsilon$ values $(0.0025,0.005,0.01,0.02,0.04)$ and the analytical solutions. The $L^2$ norms of the relative differences are defined as:
\begin{align}
E_{T}=\frac{\|T-\bar{T}\|_{L^{2}}}{\|\bar{T}\|_{L^{2}}} ~~~ {\rm and} ~~~E_{\Phi}=\frac{\|\Phi-\bar{\Phi}\|_{L^{2}}}{\|\bar{\Phi}\|_{L^{2}}}
\end{align}
for the temperature and stream-function, respectively, where $\Phi$ represents the numerical result of the stream-function. The numerical results are shown in Tab.~\ref{relative_error_different_eps}, indicating a decrease in the $L^2$ norm of the relative differences as the value of $\epsilon$ decreases for both the temperature field and the stream-function.

\begin{figure}[ht]
\centering
\includegraphics[scale=0.6]{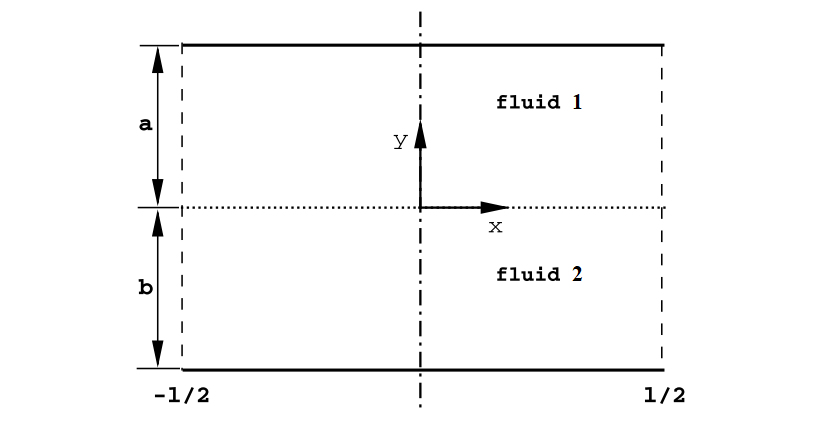}
\caption{The diagram depicting two immiscible fluids in a microchannel. The temperatures of the lower and upper walls are given by ${T(x,-b)=T_{h}+}{T_{0} \cos (\omega x)}$ and ${T(x, a)=T_{c}}$, respectively, where ${T_{h}>T_{c}>T_{0}}$ and ${\omega=\frac{2 \pi}{l}}$ is the wave number. Refer to \S\ref{Thermocapillary-convection-of-two-planar-fluid} for further details.} \label{diagram_planar_fluids}
\end{figure}

%\begin{flushleft}
\begin{figure}[H]
\subfigure[]{\includegraphics[scale=0.15]{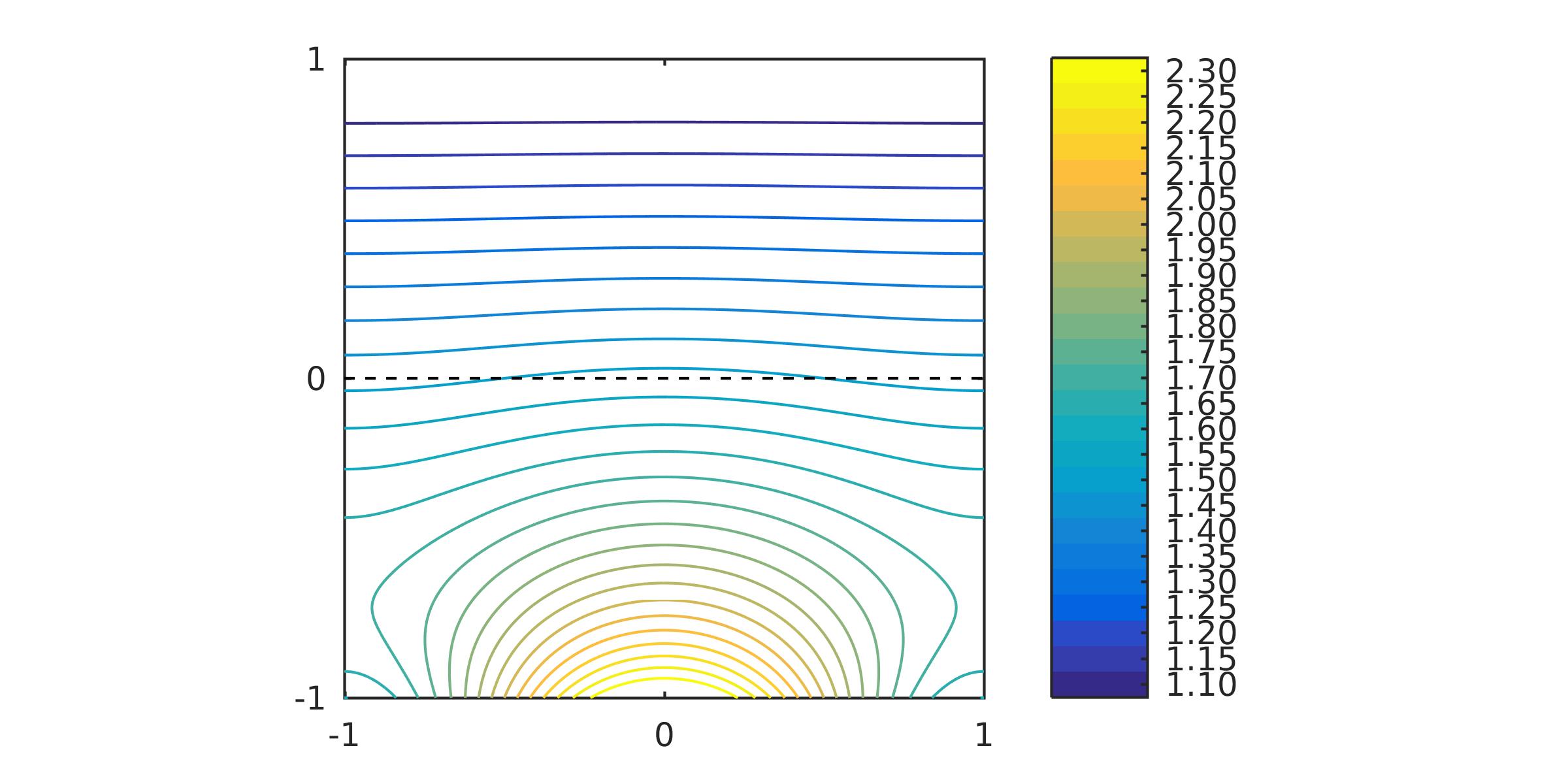}}\\%%\label{fig:25:a}//
%\hspace{-5cm}
\subfigure[]{\includegraphics[scale=0.15]{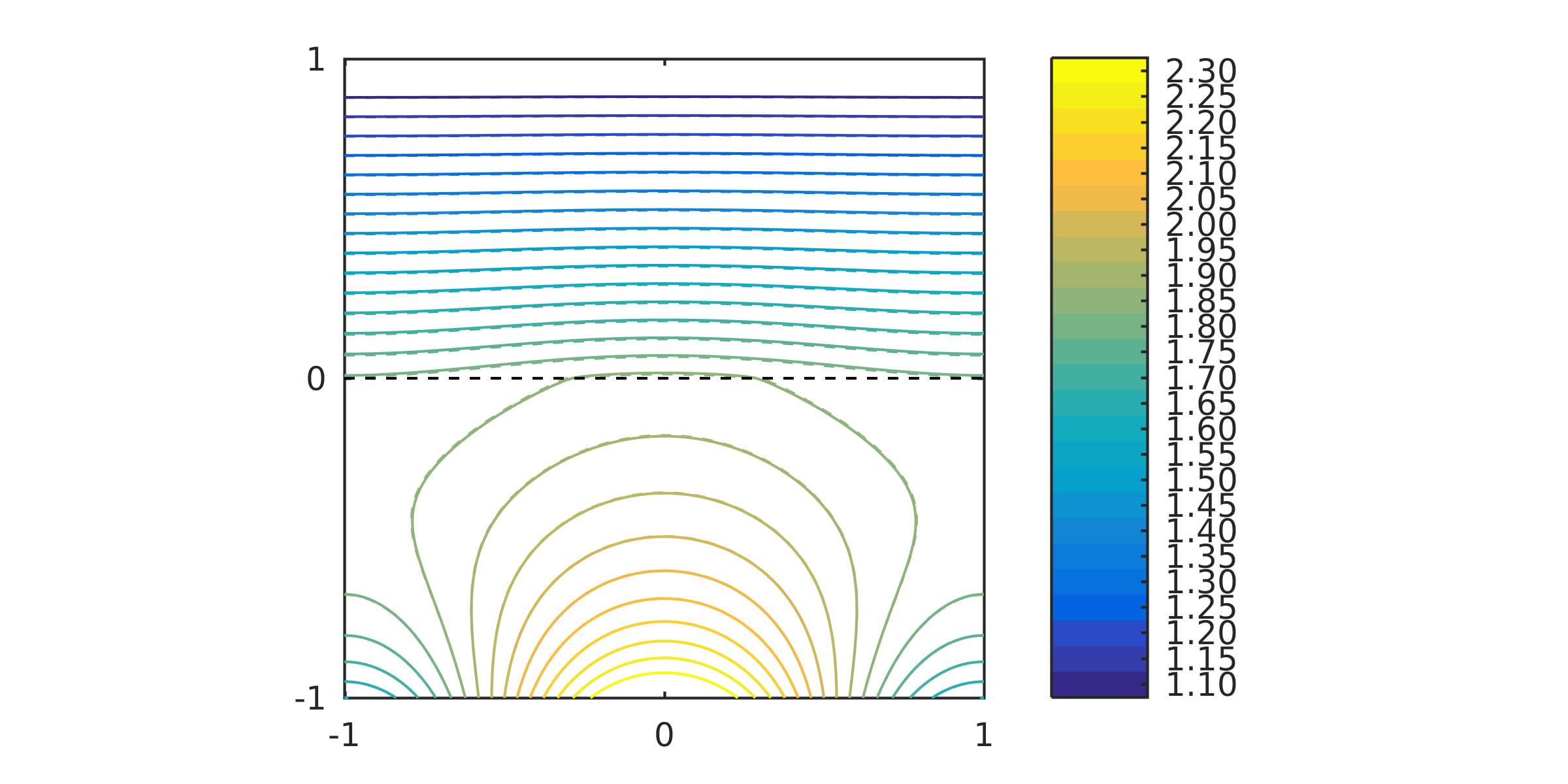}}
\caption{Isotherms for fluid systems with thermal conductivity ratios of (a) $\tilde{k}=1$ and (b) $\tilde{k}=0.2$. Solid lines represent numerical results $T$, while dashed lines depict analytical solutions $\bar{T}$. For further details, refer to \S\ref{Thermocapillary-convection-of-two-planar-fluid}.}
 \label{T_countor}
\end{figure}
%\end{flushleft}

\begin{figure}[ht]
\centering
\subfigure[]{\includegraphics[scale=0.125]{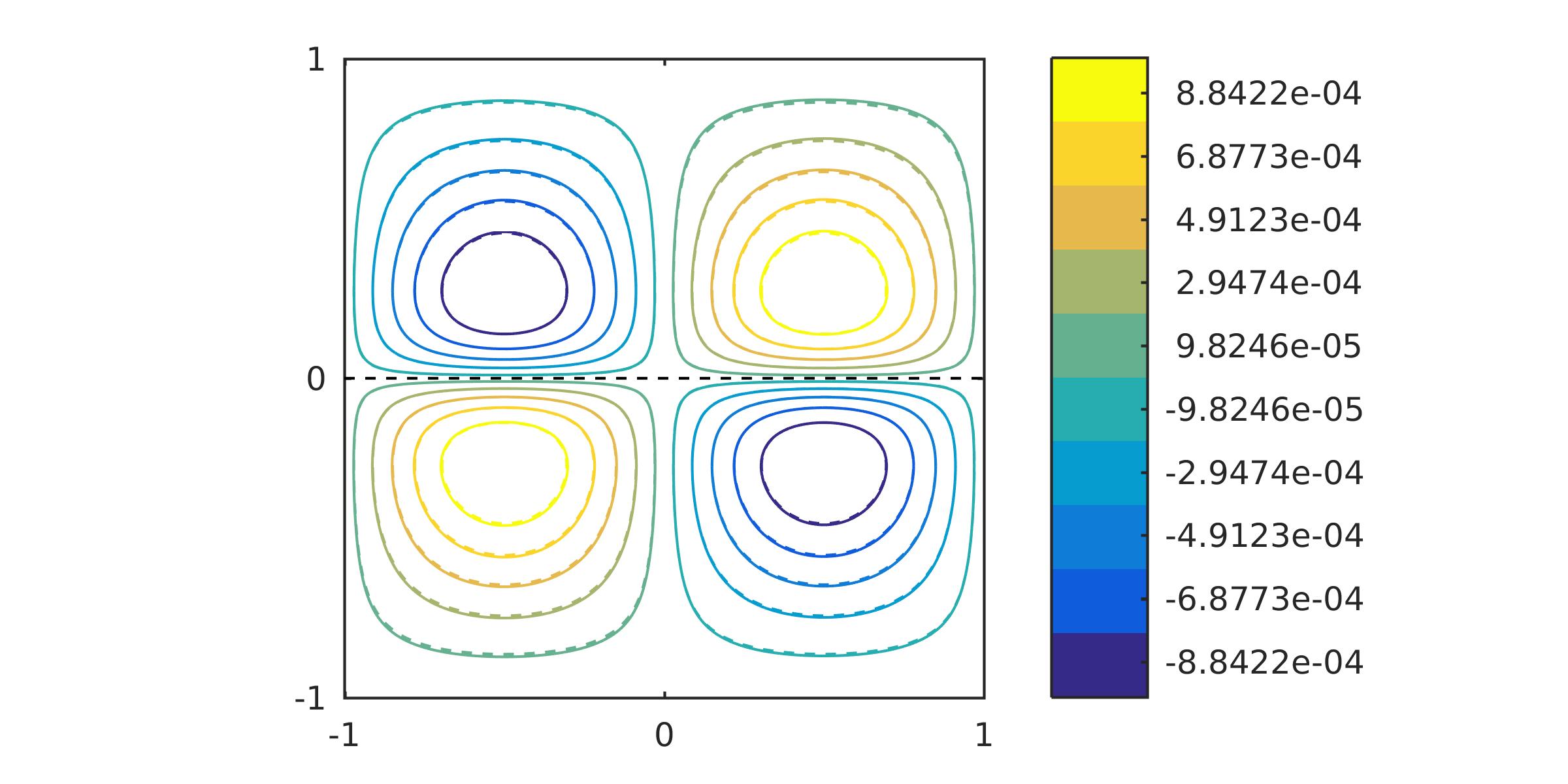}}\\%%\label{fig:25:a}//
\subfigure[]{\includegraphics[scale=0.125]{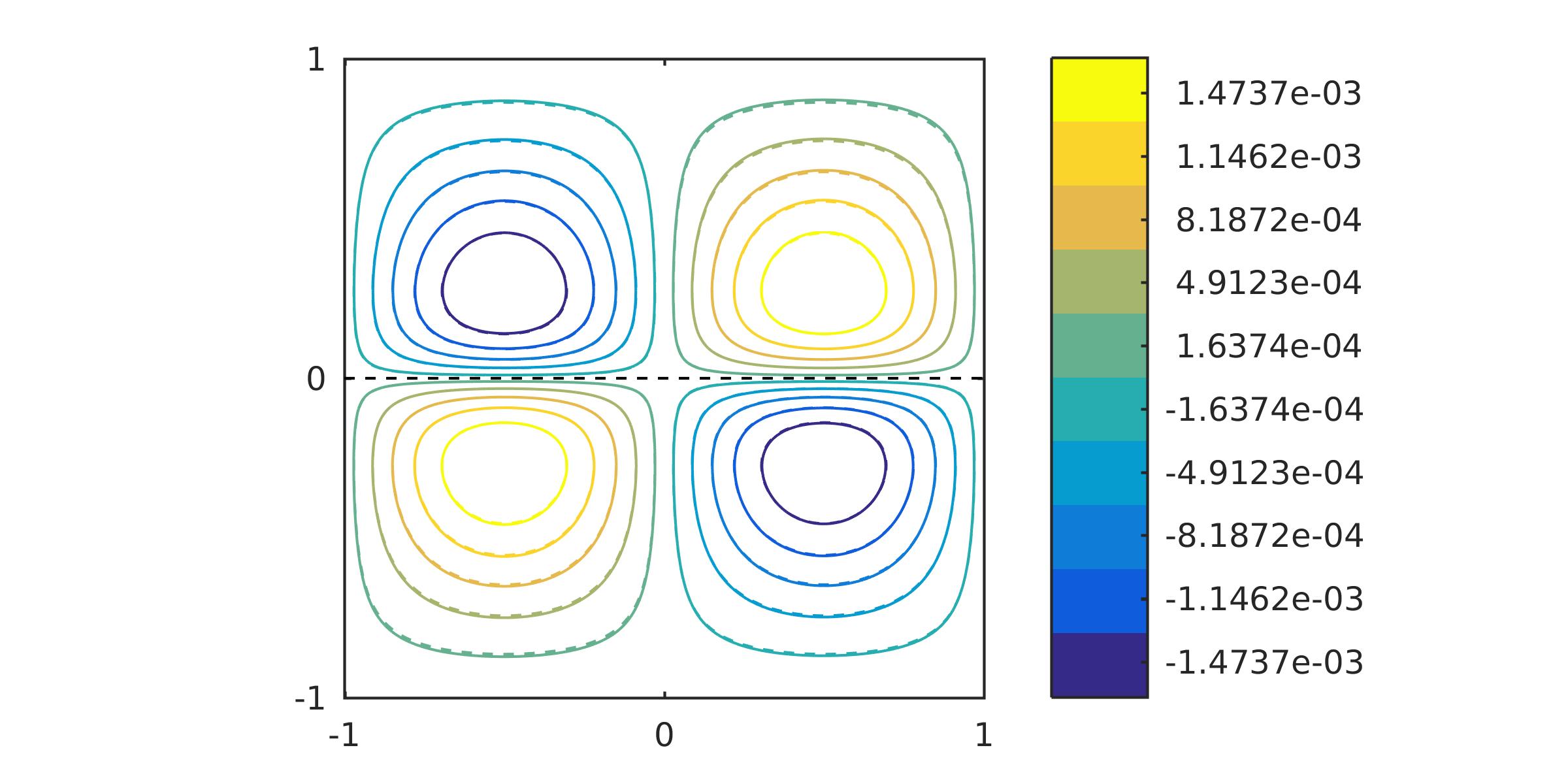}}
\caption{Streamlines for fluid systems with different thermal conductivity ratios: (a) $\tilde{k}=1$ and (b) $\tilde{k}=0.2$. Solid lines represent numerical results $\Phi$, while dashed lines depict analytical solutions $\bar{\Phi}$. For further details, refer to \S\ref{Thermocapillary-convection-of-two-planar-fluid}.}
\label{stramline}
\end{figure}

%%%%%%%%%%%%%%%%%%%%%%%%%%%%%%%%%%%%%%%%%%%%%%%%%%%%%%%%%%
\begin{table}[ht]
\centering
\caption{The $L^2$ norms of the relative differences between the numerical results obtained at the steady state and analytical solutions with various $\epsilon$ and a fixed grid size $1024\times1024$. For more details, refer to \S\ref{Thermocapillary-convection-of-two-planar-fluid}. \label{relative_error_different_eps}}
\renewcommand\arraystretch{1.25}
\begin{tabular}{ccccccc}
\hline\hline & & \multicolumn{2}{c}{ $\tilde{k}=1$ } & & \multicolumn{2}{c}{ $\tilde{k}=0.2$} \\
\cline { 3 - 4 } \cline { 6 - 7 } $\epsilon$ & & $E_T$ & $E_{\Phi}$ & & $E_T$ & $E_{\Phi}$ \\
\hline
$0.04$ & & 1.8290${\rm e}$-07 & 1.5493${\rm e}$-01 & & 1.0591${\rm e}$-02 & 1.6071${\rm e}$-01 \\
$0.02$ & & 1.8290${\rm e}$-07 & 4.4292${\rm e}$-02 & & 5.7621${\rm e}$-03 & 6.0471${\rm e}$-02 \\
$0.01$ & & 1.8290${\rm e}$-07 & 1.6597${\rm e}$-02 & & 2.9972${\rm e}$-03 &  2.3082${\rm e}$-02  \\
$0.005$ & & 1.8290${\rm e}$-07 & 1.2590${\rm e}$-02 & & 1.5303${\rm e}$-03 & 1.0508${\rm e}$-02  \\
$0.0025$ & & 1.8290${\rm e}$-07 & 1.2233${\rm e}$-02 & & 7.7892${\rm e}$-04 & 9.1579${\rm e}$-03 \\
\hline \hline
\end{tabular}
\end{table}

\subsection{The merging process of two bubbles with thermocapillary effects} \label{two-bubble-rise}
We now investigate the merging process of two spherical gas bubbles in a squared channel with an initially linear temperature field imposed along the channel (${\partial T}/{\partial z}=\nabla T_{\infty}$). The two bubbles, surrounded by a viscous liquid, are initially stationary with the same radius $R_{0}$. We adopt $R^*=R_{0}$, $V^* =\sqrt{gR_{0}}$, and $T^{*}=|\nabla T_{\infty}|R_{0}$ as the characteristic length, velocity, and temperature, respectively. The dimensionless computational domain is $[0, 8] \times [0, 8]\times [0, 16]$, and the initial bubbles, each with a radius of 1, are centered at $(4,4,5)$ and $(5.6, 4,2)$, respectively. The following initial conditions are imposed for the numerical simulations:
\begin{align}
\boldsymbol{v}(x,y,z,0)&= 0,\label{initial-v-two-bubble}\\
\psi(x,y,z,0)& = 1+\frac{1}{2}\tanh{\left(\frac{1-\sqrt{r_{1}}}{2\sqrt{2}\epsilon}\right)}+ \frac{1}{2}\tanh{\left(\frac{1-\sqrt{r_{2}}}{2\sqrt{2}\epsilon}\right)},\\
T(x,y,z,0)&=z+20,~~~~~~({\rm{case~1}}) \label{increasing_temperature}\\
{\rm{or}}~~~~~~~~~~~~~~~\notag\\
T(x,y,z,0)&=-z+36,~~~~({\rm{case~2}})\label{decreasing_temperature}
\end{align}
where $r_{1}= (x-4)^2+(y-4)^2+(z-5)^2$ and $r_{2}= (x-5.6)^2+(y-4)^2+(z-2)^2$. No-slip boundary conditions are imposed for $\boldsymbol{v}$ on all domain boundaries, while no-flux boundary conditions are applied for $\psi$ and $\mu_{c}$. Additionally, no-flux boundary conditions are set for $T$ on all side boundaries ($x=0,8$,~$y=0,8$), and $T$ is specified on the top and bottom boundaries ($z=0,16$) through (\ref{increasing_temperature}) or (\ref{decreasing_temperature}). The model parameters and ratios of physical properties are given as follows
%The dimensional computational domain is  $[0, 0.04] \times [0, 0.04]\times[0,0.08]$, and the initial bubble radius $R_{0}$ is 0.005.
\begin{align}
&\mathrm{Pe}_{\psi}=40, \mathrm{Re}=0.34, {\rm We}=0.05, {\rm Ca}=0.14, {\rm Ma}=0.2, {\rm Ec}=0.001, {\rm Fr}=1,\notag\\
& \zeta_{\rho}=1000, \zeta_{\mu}=100, \zeta_{C_h}=4, \zeta_{k}=23,\epsilon=0.03,\eta=6 \sqrt{2}.
\end{align}%\mathrm{Pe}_{T}=42.5691,
\subsubsection{The merging process of two bubbles under isothermal conditions}  \label{isothermal conditions}
We first investigate the merging behavior of two bubbles under isothermal conditions. Here, surface tension is assumed to be constant ($\sigma=1$) over the bubble interfaces, and the heat equation (\ref{sys-energy-nodim}) is dropped. The numerical results, as shown in Fig.~\ref{numerical_experimental}, are found to align well with experimental results \cite{brereton1991coaxial}. In addition, it is found that the bubbles gradually rise and deform due to the buoyancy force. As two bubbles approach, the lower bubble accelerates because the drag force on lower bubble becomes smaller than the upper bubble.
\begin{figure}[ht]
\centering
\includegraphics[scale=0.15]{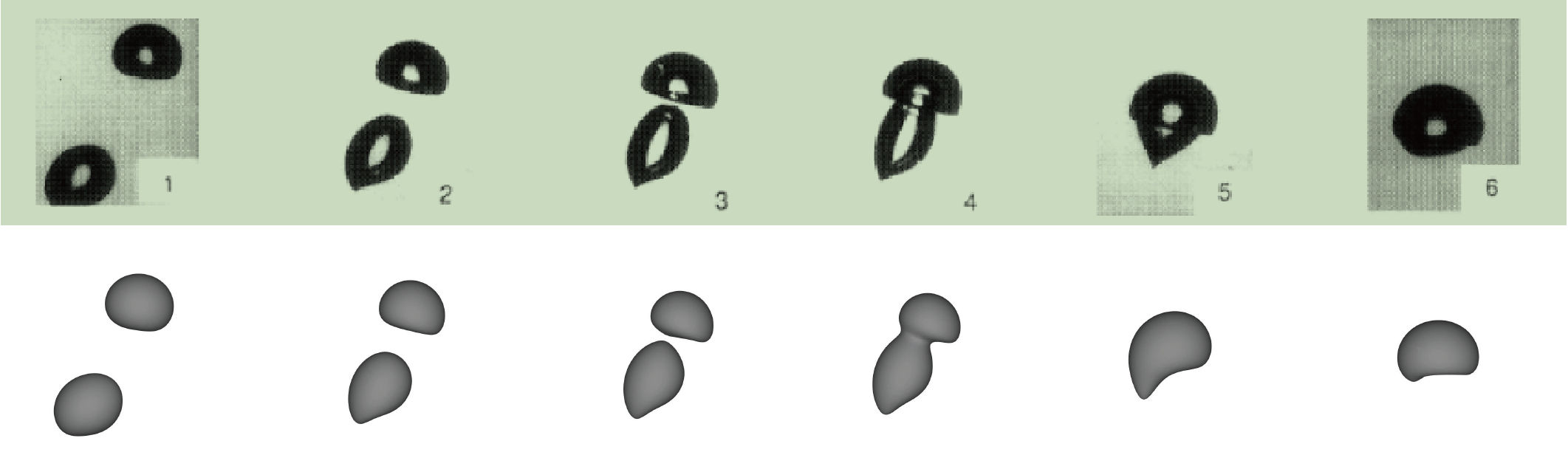}%%\label{fig:25:a}//
\caption{Comparison of numerical results and experimental results for merging process of two bubbles, adapted from \cite{brereton1991coaxial}. Refer to \S\ref{isothermal conditions} for further details.}
\label{numerical_experimental}
\end{figure}

\subsubsection{Effect of thermocapillarity on the merging of two bubbles}\label{non-isotherm-condition}
We further investigate the merging behavior of two bubbles under non-isothermal conditions by considering the full model Eqs.~\eqref{sys-vel-nodim}-\eqref{sys-energy-nodim}. Two cases of initial temperature, given by Eqs.~(\ref{increasing_temperature}) and (\ref{decreasing_temperature}), are examined to analyze the effect of the heat Peclect number $\mathrm{Pe}_{T}$ on the merging process of bubbles. Various $\mathrm{Pe}_{T}$, specially $\mathrm{Pe}_{T}=4$, $16$, $64$, $256$, are selected for this examination, with other settings described in \S\ref{two-bubble-rise}. The numerical results are illustrated in Figs.~\ref{increasing} and \ref{decreasing}.

In case 1, where the initial temperature field increases linearly along the channel, the numerical results are presented in Fig.~\ref{increasing}. Comparing with the results under isothermal conditions, several observations are made. Firstly, prior to bubble merging, the upper bubble becomes flatter and the lower bubble elongates. Following the merging, the resulting bubble also flattens further. Secondly, the presence of a positive temperature gradient causes bubbles to merge earlier, evident from the second column of Fig.~\ref{increasing}. Lastly, as $\mathrm{Pe}_{T}$ increases, the merging process slows down and the isotherms surrounding the bubbles become more distorted.

In case 2, where the initial temperature field decreases linearly long the channel, the numerical results are displayed in Fig.~\ref{decreasing}. Similar to case 1, prior to merging, the upper bubble flattens while the lower bubble elongates, followed by further flattening of the resulting bubble post-merging. Moreover, the negative temperature gradient accelerates bubble merging, noticeable from the second column of Fig.~\ref{decreasing}. Furthermore, as $\mathrm{Pe}_{T}$ increases, bubble merging occurs more rapidly, accompanied by greater distortion of the surrounding isotherms.

\begin{figure}[H]
\centering
\includegraphics[scale=0.33]{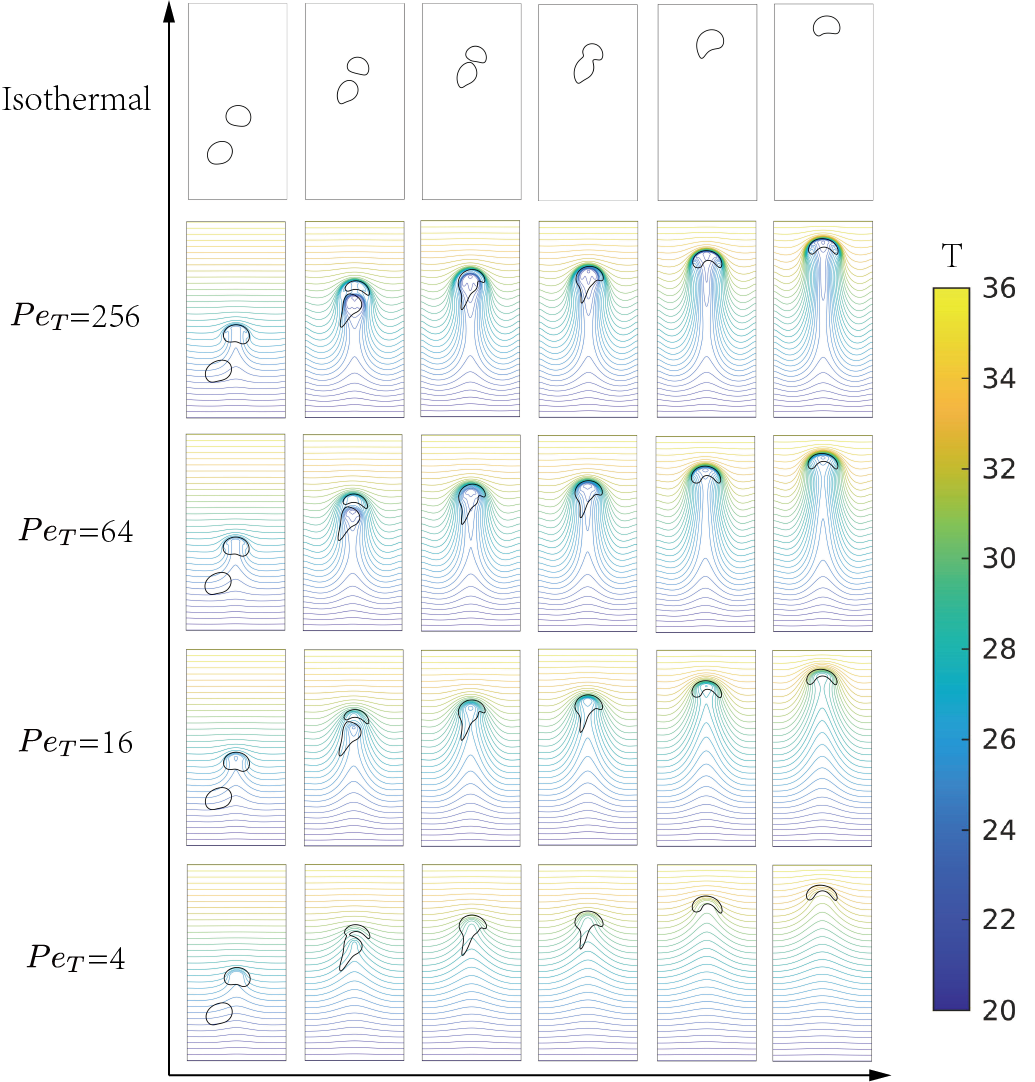}%%\label{fig:25:a}//
\caption{Time evolution of bubbles interfaces (contour $\psi=0.5$) on the slice $y=4$ under isothermal conditions, along with the time evolution of isotherms and bubbles interface on the same slice at different $\mathrm{Pe}_{T}$ numbers when the initial temperature field increases linearly along the channel. See \S\ref{non-isotherm-condition} for details.}
\label{increasing}
\end{figure}

\begin{figure}[H]
\centering
\includegraphics[scale=0.33]{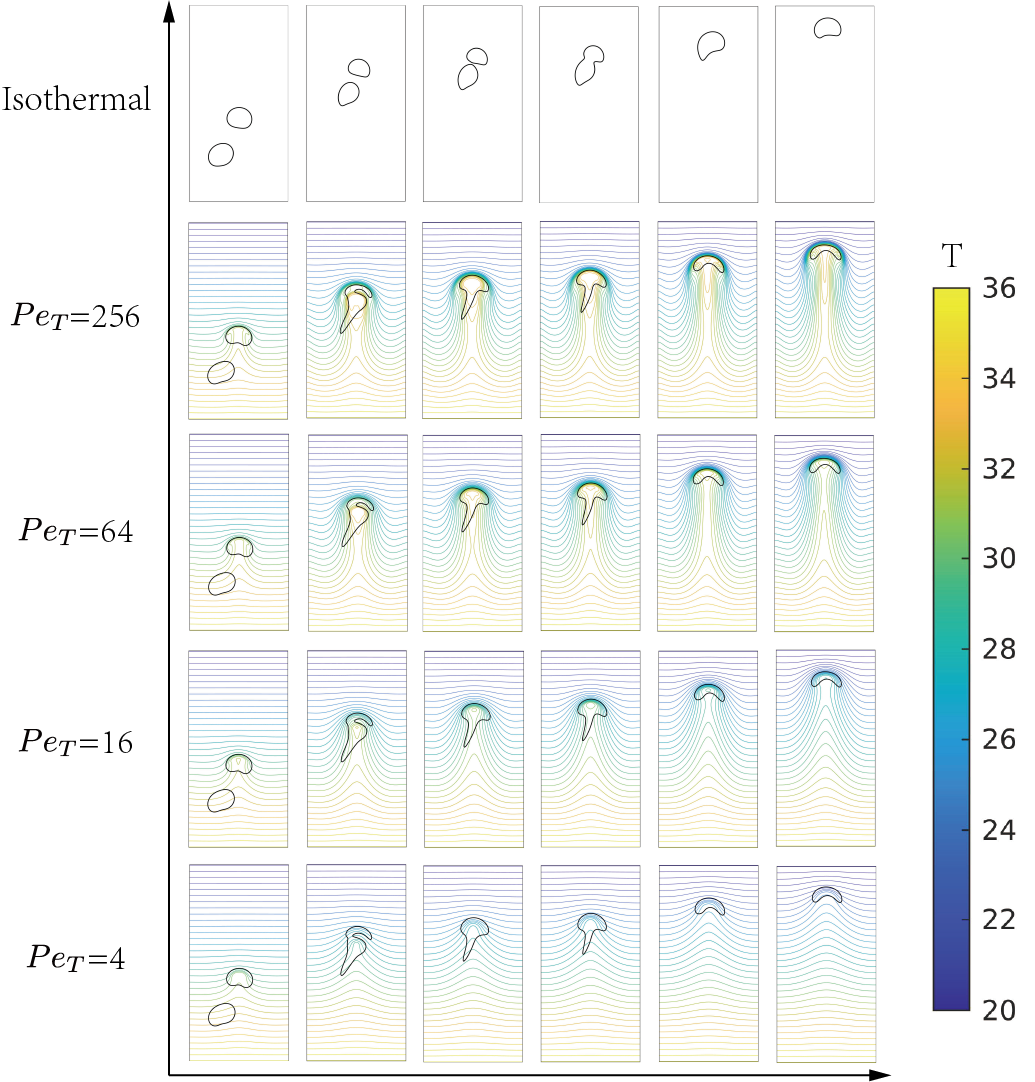}%%\label{fig:25:a}//
\caption{Time evolution of bubbles interfaces (contour $\psi=0.5$) on the slice $y=4$ under isothermal conditions, along with the time evolution of isotherms and bubble interfaces on the same slice at different $\mathrm{Pe}_{T}$ numbers when the initial temperature field decreases linearly along the channel. See \S\ref{non-isotherm-condition} for details.}
\label{decreasing}
\end{figure}

\section{Conclusion}\label{conclusion}
In this study, we present a thermodynamically consistent phase-field model for two-phase flows with thermocapillary effects, which allows the two fluid components to have different physical properties, including density, viscosity, heat capacity, and thermal conductivity, and meanwhile maintains the balance laws of mass, momentum, and energy and entropy increase simultaneously. 

Given the highly coupled and nonlinear nature of the model equations, we develop a first order accurate numerical method that satisfies the discrete laws of mass conservation and entropy increasing. We validate both our model and numerical method through a series of numerical tests. These include the thermocapillary migration of a droplet, thermocapillary convection within a heated microchannel featuring two superimposed planar fluids, and the merging dynamics of two bubbles under isothermal conditions. Remarkably, our results align closely with existing analytical solutions or experimental findings. Furthermore, we delve into the merging process of bubbles under non-isothermal conditions. In comparing these scenarios with those under isothermal conditions, we observe significant differences: temperature gradients prompt earlier bubble merging and substantial alterations in bubble morphology. Moreover, as the initial temperature field increases linearly along the channel, the merging rate of bubbles decelerates with rising heat Peclect number ($\mathrm{Pe}_{T}$). Conversely, for scenarios with a linear decrease in the initial temperature field along the channel, bubbles coalesce at a faster rate as $\mathrm{Pe}_{T}$ increases.

\section*{Acknowledgement}
This work is partially supported by National Natural Science Foundation of China (Grant No. 12371387 and No. 12071046), Natural Science Foundation of Shandong Province (Grant No. ZR2021QA018) and China Postdoctoral Science Foundation
Funded Projection (Grant No. 2022M721757).
\newpage
\bibliographystyle{elsarticle-num}%%%%%
\bibliography{LRe1}

    \clearpage
    \newpage
    \appendix

%\renewcommand\theequation{D\arabic{equation}}

%\renewcommand\theequation{E\arabic{equation}}
%The entropy conservation Equation~\eqref{def-entropy-laws} states that the rate of change of entropy in the control volume during the process equals the net entropy transfer through the boundary (classical ${\boldsymbol{q}_{E}}/{T}$ and non-classical $\boldsymbol{q}^{nc}_{S}$ entropy flux) plus the local entropy generation ($s_{gen}\ge 0$) within the control volume. Based on the second law of thermodynamics, the local entropy generation is non-negative for a dissipative system (or, say, for an irreversible process), which is the key to the thermodynamic frame that we used for the derivations.
%\renewcommand\theequation{F\arabic{equation}}
\section{Derivation of the conservation equations}\label{derivation_of_coservation_equa}
\subsection{Mass and volume conservation equations}
Substituting Eq.~(\ref{def-mass}) into Eq.~(\ref{def-mass-laws}), and substituting Eq.~(\ref{def-phase}) into Eq.~(\ref{def-phase-laws}), we obtain the following mass conservation equation
\begin{align}\label{mass-conservation-equa}
	\frac{\partial \rho}{\partial t}+\boldsymbol{\nabla} \cdot(\rho \boldsymbol{v})&=0,
\end{align}
and volume conservation equation
\begin{align} \label{component_1_volume_conservation}
 \frac{\partial \psi}{\partial t}+\boldsymbol{\nabla} \cdot(\psi \boldsymbol{v})&=-\frac{\boldsymbol{\nabla} \cdot \boldsymbol{J}}{\rho_{1}}.
\end{align}
Similarly, we assume the volume conservation equation for fluid 2 is
\begin{align}\label{component_2_volume_conservation}
	\frac{\partial \Phi_{2}}{\partial t}+\boldsymbol{\nabla} \cdot\left(\Phi_{2} \boldsymbol{v}\right)=-\frac{\boldsymbol{\nabla} \cdot \tilde{\boldsymbol{J}}}{\rho_{2}},
\end{align}
where $\Phi_{2}$ is the volume fraction of fluid 2, and $\tilde{\boldsymbol{{J}}}$ is the volume flux of fluid 2.
Multiplying Eq.~(\ref{component_1_volume_conservation}) by $\rho_{1}$ and Eq.~(\ref{component_2_volume_conservation}) by $\rho_{2}$, adding them together, and using Eq.~(\ref{mass-conservation-equa}), we obtain
\begin{align}
	-\boldsymbol{\nabla} \cdot \boldsymbol{J}-\boldsymbol{\nabla} \cdot \tilde{\boldsymbol{J}}=0
\end{align}
or
\begin{equation*}
	\begin{split}
	\boldsymbol{\nabla} \cdot\left(\boldsymbol{J}+\tilde{\boldsymbol{J}}\right)=0.
\end{split}
\end{equation*}
Furthermore, adding Eqs.~(\ref{component_1_volume_conservation}) and (\ref{component_2_volume_conservation}) together, we obtain
\begin{align}\label{quasi_incompressible_equation}
	\boldsymbol{\nabla} \cdot \boldsymbol{v} & =-\frac{\boldsymbol{\nabla} \cdot \boldsymbol{J}}{\rho_{1}}-\frac{\boldsymbol{\nabla} \cdot \tilde{\boldsymbol{J}}}{\rho_{2}} =\frac{\rho_{1}-\rho_{2}}{ \rho_{1}\rho_{2}} \boldsymbol{\nabla} \cdot \boldsymbol{J}.
\end{align}
Let $\alpha=(\rho_{2}-\rho_{1})/{ \rho_{2}}$, thus Eq.~(\ref{quasi_incompressible_equation}) can be rewritten as
\begin{align}
	\boldsymbol{\nabla} \cdot \boldsymbol{v}=-\alpha \frac{\boldsymbol{\nabla} \cdot \boldsymbol{J}}{\rho_{1}}. \label{Quasi_incompressible_0}
\end{align}
Recall the definition of the total time derivative:
\begin{align} \label{material derivative}
\frac{D \psi}{D t}=\frac{\partial \psi}{\partial t}+\boldsymbol{v} \cdot \boldsymbol{\nabla} \psi,
\end{align}
then (\ref{component_1_volume_conservation}) can be rewritten as
	\begin{align} \label{mass_conserve_0}
\frac{D \psi}{D t}+\psi(\boldsymbol{\nabla} \cdot \boldsymbol{v})=\frac{-\boldsymbol{\nabla} \cdot \boldsymbol{J}}{\rho_{1}}.
\end{align}

\subsection{Momentum and energy conservation equations}
Substituting Eq.~(\ref{def-mom}) into Eq.~(\ref{def-mom-laws}),
and using Eq.~(\ref{mass-conservation-equa}), we obtain the momentum conservation equation
\begin{align}\label{momentum_equation_0}
	\rho \frac{D \boldsymbol{v}}{D t}=\boldsymbol{\nabla} \cdot \mathbf{T}-\rho g \hat{\boldsymbol{z}}.
\end{align}
In addition, substituting Eq.~(\ref{def-energy}) into Eq.~(\ref{def-energy-laws}), we obtain
\begin{align} \label{energy_conser_equa_0-1}
\frac{\partial \hat{u}}{\partial t}+\boldsymbol{\nabla} \cdot(\hat{u}\boldsymbol{v})=-\boldsymbol{\nabla} \cdot \boldsymbol{q}_{E}+\boldsymbol{\nabla} \boldsymbol{v}: \mathbf{T}.
\end{align}
Using the definition of $\hat{u}$ (in \ref{definition_hat_u}), we obtain
\begin{align} \label{energy_conser_equa_1}
\frac{\partial \hat{u}}{\partial t}=&\frac{\partial u}{\partial T} \frac{\partial T}{\partial t}+\frac{\partial u}{\partial \psi} \frac{\partial \psi}{\partial t}+\lambda_{u}\frac{\partial \delta}{\partial t}.
\end{align}
Substituting the definition of $\delta$ (\ref{definition_of_delta}) into (\ref{energy_conser_equa_1}), we obtain
\begin{align} \label{energy_conser_equa_2}
\frac{\partial \hat{u}}{\partial t}&=\frac{\partial u}{\partial T} \frac{\partial T}{\partial t}+\frac{\partial u}{\partial \psi} \frac{\partial \psi}{\partial t}+\lambda_{u}\frac{ W^{\prime}(\psi)}{\epsilon}\frac{\partial \psi}{\partial t} \notag\\
&+\lambda_{u} \epsilon \boldsymbol{\nabla} \psi \cdot \boldsymbol{\nabla} \left(\frac{\partial \psi}{\partial t}\right) \notag\\
&= \frac{\partial u}{\partial T} \frac{\partial T}{\partial t}+\frac{\partial u}{\partial \psi} \frac{\partial \psi}{\partial t}+\lambda_{u}\frac{ W^{\prime}(\psi)}{\epsilon}\frac{\partial \psi}{\partial t} \notag\\
+&\lambda_{u} \epsilon \boldsymbol{\nabla} \cdot (\boldsymbol{\nabla} \psi \frac{\partial \psi}{\partial t})-\lambda_{u} \epsilon \Delta \psi\frac{\partial \psi}{\partial t} \notag\\
&= \frac{\partial u}{\partial T} \frac{\partial T}{\partial t}+\left(\frac{\partial u}{\partial \psi}+\lambda_{u}w\right) \frac{\partial \psi}{\partial t}+\lambda_{u} \epsilon \boldsymbol{\nabla} \cdot (\boldsymbol{\nabla} \psi \frac{\partial \psi}{\partial t} ),
\end{align}
where
\begin{align}
w=\frac{W^{\prime}(\psi)}{\epsilon} -\epsilon \Delta \psi.
\end{align}
In addition, we have the following identity:
%With the help of the following identity:
\begin{align} \label{nabla_u}
 &\boldsymbol{v}\cdot \boldsymbol{\nabla} \hat{u} = \boldsymbol{v} \cdot\left(\frac{\partial u}{\partial T} \boldsymbol{\nabla} T+\left(\frac{\partial u}{\partial \psi}+\lambda_{u} w \right) \boldsymbol{\nabla} \psi +\lambda_{u} \epsilon \boldsymbol{\nabla}\cdot (\boldsymbol{\nabla} \psi\otimes \boldsymbol{\nabla} \psi)\right).
\end{align}
Substituting (\ref{energy_conser_equa_2}) and (\ref{nabla_u}) into (\ref{energy_conser_equa_0-1}), we obtain the energy conservation equation
\begin{align} \label{energy_equation}
\frac{\partial u}{\partial T} \frac{D T}{D t}&=-\left(\frac{\partial u}{\partial \psi}+\lambda_{u} w\right) \frac{D \psi}{D t}-\lambda_{u} \epsilon \boldsymbol{\nabla} \cdot (\boldsymbol{\nabla} \psi \frac{\partial \psi}{\partial t})-\hat{u}(\boldsymbol{\nabla} \cdot \boldsymbol{v}) \notag\\
&-\lambda_{u} \epsilon \boldsymbol{\nabla}\cdot (\boldsymbol{\nabla} \psi\otimes \boldsymbol{\nabla} \psi)\cdot \boldsymbol{v}-\boldsymbol{\nabla} \cdot \boldsymbol{q}_{E}+\boldsymbol{\nabla} \boldsymbol{v}: \mathbf{T}.
\end{align}

%\section{Proof for the continuous conservation laws}
%\subsection{Proof for the continuous mass and energy conservation laws}\label{proof_continous_energy law}

%\subsection{Proof for the continuous law of entropy increase}\label{proof_continous entropy increase}

%\section*{References}

%\bibliographystyle{elsarticle-num}%%%%%
%\bibliography{LRe1}

\end{document}